\documentclass[journal]{IEEEtran}

\ifCLASSINFOpdf
\usepackage[pdftex]{graphicx}
\DeclareGraphicsExtensions{.pdf,.jpeg,.png}
\else
\usepackage[dvips]{graphicx}
\fi
\usepackage{epstopdf}
\usepackage{amssymb,amsmath,amsthm}
\usepackage{amssymb}
\usepackage{wasysym}
\usepackage{algorithm}
\usepackage{algorithmic}
\usepackage{array}
\usepackage{cite}
\usepackage{color}
\usepackage{url}

\usepackage{epsfig,latexsym}
\usepackage{flushend}
\usepackage{verbatim}
\usepackage{amsopn}
\usepackage{booktabs}

\usepackage{stfloats}
\usepackage{enumerate}
\usepackage{hyperref}
\usepackage{subfigure}
\usepackage{caption}
\usepackage{bm}
\newtheorem{theorem}{Theorem}

\renewcommand{\baselinestretch}{0.95}

\begin{document}
	
	\title{Movable Antenna Enhanced Wireless Sensing Via Antenna Position Optimization}
	
	\author{{Wenyan Ma, \IEEEmembership{Graduate Student Member, IEEE}, Lipeng Zhu, \IEEEmembership{Member, IEEE}, and  Rui Zhang, \IEEEmembership{Fellow, IEEE}}
		\thanks{W. Ma and L. Zhu are with the Department of Electrical and Computer Engineering, National University of Singapore,
			Singapore 117583 (Email: {wenyan@u.nus.edu}, {zhulp@nus.edu.sg}).
			
			R. Zhang is with School of Science and Engineering, Shenzhen Research Institute of Big Data, The Chinese University of Hong Kong, Shenzhen, Guangdong 518172, China (e-mail: rzhang@cuhk.edu.cn). He is also with the Department of Electrical and Computer Engineering, National University of Singapore, Singapore 117583 (e-mail: elezhang@nus.edu.sg).
	}}
	\maketitle
	
	\begin{abstract}
	In this paper, we propose a new wireless sensing system equipped with the movable-antenna (MA) array, which can flexibly adjust the positions of antenna elements for improving the sensing performance over conventional antenna arrays with fixed-position antennas (FPAs). First, we show that the angle estimation performance in wireless sensing is fundamentally determined by the array geometry, where the Cramer-Rao bound (CRB) of the mean square error (MSE) for angle of arrival (AoA) estimation is derived as a function of the antennas' positions for both one-dimensional (1D) and two-dimensional (2D) MA arrays. Then, for the case of 1D MA array, we obtain a globally optimal solution for the MAs' positions in closed form to minimize the CRB of AoA estimation MSE. While in the case of 2D MA array, we aim to achieve the minimum of maximum (min-max) CRBs of estimation MSE for the two AoAs with respect to the horizontal and vertical axes, respectively. In particular, for the special case of circular antenna movement region, an optimal solution for the MAs' positions is derived under certain numbers of MAs and circle radii. Thereby, both the lower- and upper-bounds of the min-max CRB are obtained for the antenna movement region with arbitrary shapes. Moreover, we develop an efficient alternating optimization algorithm to obtain a locally optimal solution for MAs' positions by iteratively optimizing one between their horizontal and vertical coordinates with the other being fixed. Numerical results demonstrate that our proposed 1D/2D MA arrays can significantly decrease the CRB of AoA estimation MSE as well as the actual MSE compared to conventional uniform linear arrays (ULAs)/uniform planar arrays (UPAs) with different values of uniform inter-antenna spacing. Furthermore, it is revealed that the steering vectors of our designed 1D/2D MA arrays exhibit low correlation in the angular domain, thus effectively reducing the ambiguity of angle estimation.
		
	\end{abstract}
	
	\section{Introduction}
	
	The forthcoming sixth-generation (6G) mobile communication systems are expected to support more location-aware applications, such as autonomous driving, robot navigation, and virtual reality \cite{jiang2021the,saad2020a,chowdhury20206g}. These applications demand for improved sensing capabilities of wireless networks, beyond the conventional  quality of service (QoS) requirement for data transmission rate and reliability. Consequently, there is a growing interest in the new paradigm of integrated sensing and communication (ISAC) in future wireless networks, which are able to provide both sensing and communication services by sharing the hardware and radio resources. It is expected that wireless sensing, which generally encompasses the detection, estimation, and/or extraction of physical information from environmental targets, will become a new major service offered by 6G wireless networks \cite{liu2022survey,shao2024intelligent}.
	
	To synthesis sensing beams with high angular resolution and beamforming gain, large-scale antenna arrays are usually required for sensing nodes such as radar and base station (BS) \cite{mailloux2005phased,wirth2005radar}. However, their hardware cost and power consumption also scale up with the increase of number of antennas, which brings a challenge in implementing low-cost and high-performance sensing systems. 
	To reduce the implementation cost, sparse antenna arrays were proposed to reduce the number of antennas with enlarged inter-antenna spacing for achieving a large portion of the angular resolution of large-scale antenna arrays \cite{greene1978sparse,roberts2011sparse}. However, sparse arrays usually adopt fixed-position antennas (FPAs), which cannot adapt to different sensing requirements in wireless networks and switch between the optimal array geometries for sensing and communication applications \cite{wang2023can,gazzah2009optimum}. Furthermore, the FPAs in both large-scale and sparse arrays  cannot fully exploit wireless signal variation in  a given region where the sensing transmitter/receiver is located.
	
	To overcome the above limitations of wireless sensing with FPA arrays, we propose in this paper a new sensing system equipped with the movable-antenna (MA) array \cite{zhu2023movablemagzine}, which can adjust the positions of antenna elements at the sensing transmitter/receiver (also known as fluid antenna system (FAS) \cite{zhu2024historical}). Compared to conventional FPA arrays, the new degree of freedom (DoF) of MA arrays via antenna position optimization has a potential to significantly improve the sensing performance under the same number of antennas, explained as follows. First, the array aperture can be effectively increased by enlarging the antenna movement region, which helps achieve a higher resolution for angle estimation. Second, the geometry of an MA array can be optimized to decrease the correlation between steering vectors over different directions, thus reducing the ambiguity and undesired interference for angle estimation. Third, real-time adjustment of MAs' positions can adapt to time-varying environment conditions and different sensing task requirements. Moreover, the reconfiguration of MAs can support flexible functionality switching between communication and sensing in ISAC systems. In practice, the geometry of MA array can be pre-configured based on the sensing application or adjusted in real time to cater to dynamic ISAC performance requirements.
	
	Preliminary studies have demonstrated the promising benefits of MAs for wireless sensing. The authors in \cite{zhuravlev2015experi} showed a prototype of MA-aided multi-static radar, where the transmit and receive antennas can be moved in a line segment with  linear drives, such that the antenna configuration can be fine-tuned to balance the trade-off between system complexity and radar imaging quality.
	In \cite{hinske2008using}, an antenna was attached on a sliding carriage to measure the distance between it and the target at various antenna positions. The target's location was then determined by examining the overlapping regions of the coverage areas centered at these positions. However, the above studies only demonstrated that MAs can enhance sensing performance by exploiting the antenna position adjustment, while the fundamental relationship between MAs' positions and sensing performance was not revealed therein. 
	
	Additionally, existing studies have extensively explored the advantages of MAs for wireless communications. For instance, in \cite{zhu2022modeling,mei2024movable,zhu2024performance}, it was demonstrated that employing MAs can efficiently enhance the received signal-to-noise ratio (SNR) based on deterministic or stochastic channel models. In \cite{zhu2023movable,xiao2023multiuser,wu2023movable,qin2024antenna,cheng2023sum,yang2024flexible}, MA-aided multiuser communication was studied under various system configurations, where the MA position optimization helps reshape multiuser channels to save total transmit power or improve the sum-rate performance of users. In \cite{ma2022mimo,chen2023joint}, MA-aided multiple-input multiple-output (MIMO) systems were studied based on instantaneous or statistical channel state information (CSI), revealing that the joint optimization of the MAs' positions at the transmitter and receiver can significantly increase MIMO system capacity. In \cite{ma2023compressed,xiao2024channel}, channel estimation for MA systems was investigated to reconstruct the channel response between any transmit antenna and receive antenna locations by sequentially estimating the angles
	of departure (AoDs), angles of arrival (AoAs), and 	complex coefficients of multiple channel paths, where the MAs' positions were optimized to improve the AoA/AoD estimation performance by compressed sensing. Furthermore, prior works \cite{zhu2023movablebeam,ma2024multi} have demonstrated the efficacy of MA arrays in interference nulling and multi-beamforming. The studies in \cite{hu2024secure,tang2024secure} have validated the significant potential of MAs in enhancing the physical-layer security of wireless communication systems. Moreover, a new six-dimensional (6D) MA system was proposed in \cite{shao20246d,shao20246d2} to improve the long-term network capacity by jointly designing the positions and rotations of multiple antenna surfaces at the BS based on the users' spatial distribution.

	Despite the above works on MA-aided wireless sensing and communications, the fundamental performance of MA-aided wireless sensing has not been characterized in the existing literature yet. This thus motivates this paper, where we propose a new wireless sensing system equipped with the MA array, which leverages the new DoF via antenna position optimization for improving the sensing performance. To this end, we first show that the angle estimation performance in wireless sensing is crucially dependent on the array geometry, where the Cramer-Rao bound (CRB) of AoA estimation mean square error (MSE) is derived as a function of the antennas' positions for both one-dimensional (1D) and two-dimensional (2D) MA arrays. Based on these results, the main contributions of this paper are summarized as follows:
	
	\begin{itemize}
		\item First, for the case of 1D MA array, we reveal that the CRB of AoA estimation MSE is inversely proportional to the variance of the MAs' positions. Inspired by this, an optimization problem is formulated to minimize the CRB by optimizing the antenna positioning vector (APV), which constitutes all MAs' positions. Although the formulated optimization problem is non-convex, we obtain a globally optimal solution for the APV in closed form.
		\item Next, for the case of 2D MA array, we derive the CRB of AoA estimation MSE and formulate the CRB minimization problem by optimizing the 2D MAs' positions. Specifically, we aim to achieve the minimum of maximum (min-max) CRBs of estimation MSE for the two AoAs with respect to the horizontal and vertical axes, respectively. In particular, for the special case of circular antenna movement region, an optimal solution for the MAs' positions is derived under certain numbers of MAs and circle radii. Thereby, both the lower- and upper-bounds of the min-max CRB are obtained for the antenna movement region with arbitrary shapes. Moreover, we develop an efficient alternating optimization algorithm to obtain a locally optimal solution for MAs' positions by iteratively optimizing one between their horizontal and vertical coordinates with the other being fixed.
		\item Finally, we present extensive numerical results to demonstrate the decrease of AoA estimation MSE by MA arrays compared to conventional uniform linear arrays (ULAs) or uniform planar arrays (UPAs) with different values of uniform inter-antenna spacing, and shed light on how the designed antennas' positions can achieve such performance gains. In particular, it is revealed that the steering vectors of our designed 1D/2D MA arrays exhibit low correlation in the angular domain, thus effectively reducing the ambiguity of angle estimation.
	\end{itemize}
	
	The remainder of this paper is organized as follows. Section II and Section III describe the system model and analyze the CRB of AoA estimation MSE for 1D and 2D MA arrays, respectively. In Section IV, we propose an alternating optimization algorithm to solve the antenna position optimization problem for 2D MA arrays. Section V provides numerical results and pertinent discussions. Finally, we conclude this paper in Section VI.

	\textit{Notations}: Symbols representing vectors (lower case) and matrices (upper case) are in boldface. The conjugate, transpose, and conjugate transpose operations are represented by $(\cdot)^{\mathsf *} $, $(\cdot)^{\mathsf T} $, and $(\cdot)^{\mathsf H} $, respectively. The sets of $P\times{Q}$ dimensional complex and real matrices are denoted by $\mathbb{C}^{P\times{Q}}$ and $\mathbb{R}^{P\times{Q}}$, respectively. The set of positive integers is denoted by $\mathbb{Z}^{+}$. $\bm{a}[p]$ denotes the $p$th entry of vector $\bm{a}$. $\bm{I}_N$ and $\bm{1}_N$ denote the $N$-dimensional identity matrix and the $N$-dimensional column vector with all elements equal to $1$, respectively. $\mathcal{A}\setminus\mathcal{B}$ denotes the subtraction of set $\mathcal{B}$ from set $\mathcal{A}$. $\mathcal{A}\subset\mathcal{B}$ represents that $\mathcal{A}$ is a subset of $\mathcal{B}$. The ceiling and floor of real number $a$ are denoted by $\lceil a \rceil$ and $\lfloor a \rfloor$, respectively. The $2$-norm of vector $\bm{a}$ is denoted by $\|\bm{a}\|_2$.
	
	\section{CRB Characterization for 1D MA Array}
	
	\subsection{System Model}
	We consider a wireless sensing system with $N$ MAs to estimate a target's AoA as shown in Fig.~\ref{system}. With the aid of driver components, the MAs' positions can be adjusted flexibly within the given 1D line segment of length $A$. Let $x_n\in[0,A]$ represent the position of the $n$th MA ($n=1,2,\ldots,N$), and the APV of all $N$ MAs is denoted by $\bm{x}\triangleq[x_1,x_2,\ldots,x_N]^{\mathsf T}\in \mathbb{R}^{N}$, with
	$0\leq x_1 < x_2<\ldots<x_N\leq A$ without loss of generality.
	
	\begin{figure}[!t]
		\centering
		\includegraphics[width=70mm]{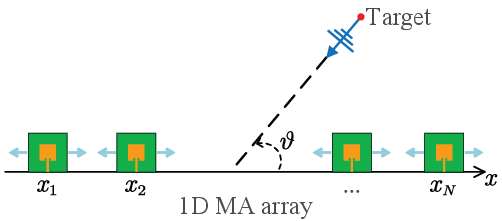}
		\caption{The 1D MA array for target angle estimation.}
		\label{system}
	\end{figure}
	
	To estimate the AoA of the target, the receiver consecutively receives the echoes reflected by the target over $T$ snapshots\footnote{The echoes are reflections of the probing signals transmitted by the sensing transmitter, which can be either co-located  or separately deployed with the sensing receiver.}. We assume that the target-receiver channel is line-of-sight (LoS), which remains static over $T$ snapshots. Since the distance between the target and receiver is typically much larger than the size of the region for antenna movement, we consider the far-field channel model from the target to receiver \cite{zhu2022modeling,zhu2023movable}. As shown in Fig.~\ref{system}, the physical steering angle of the target-receiver LoS path is denoted by $\vartheta \in [0, \pi]$. For convenience, the spatial AoA is defined as
	$u\triangleq\cos \vartheta$. The steering vector of the MA array can thus be expressed as a function of the APV $\bm{x}$ and the spatial AoA $u$, i.e.,
	\begin{equation}\label{g}
		\bm{\alpha}(\bm{x},u) \triangleq \left[ e^{j\frac{2\pi}{\lambda}x_1 u}, e^{j\frac{2\pi}{\lambda}x_2 u}, \ldots, e^{j\frac{2\pi}{\lambda}x_N u} \right]^{\mathsf T} \in{\mathbb{C}^{N}},
	\end{equation}	
	where $\lambda$ is the wavelength. Furthermore, let $\beta$ denote the complex path coefficient from the target to the origin of the line segment. As a result, the target-receiver LoS channel vector is given by
	\begin{equation}\label{H}
		\bm{h}(\bm{x},u)=\beta\bm{\alpha}(\bm{x},u).
	\end{equation}

	\subsection{AoA Estimation}
	For any given APV $\bm{x}$, we adopt the multiple signal classification (MUSIC) algorithm for
	estimating the spatial AoA $u$ of the target. Specifically, the received signal at the $t$th ($t=1,2,\ldots,T$) snapshot is given by
	\begin{equation}
		\bm{y}_t=\bm{h}(\bm{x},u)s_t + \bm{z}_t,
	\end{equation}
	where $s_t$ denotes the reflected signal from the target with average power $\mathbb{E}\{|s_t|^2\} = P$. $\bm{z}_t \sim \mathcal{CN}(0,\sigma^2\bm{I}_N)$ represents the additive white Gaussian noise (AWGN) vector at the receiver, which is assumed to be circularly symmetric complex Gaussian (CSCG) distributed with zero mean and covariance matrix $\sigma^2\bm{I}_N$, where $\sigma^2$ represents the average noise power.
	
	To estimate the spatial AoA $u$, the received signals over $T$ snapshots are stacked into the following matrix as
	\begin{equation}\label{Y}
		\bm{Y}\triangleq[\bm{y}_1,\bm{y}_2,\ldots,\bm{y}_T]=\bm{h}(\bm{x},u)\bm{s}^{\mathsf T} + \bm{Z},
	\end{equation}
	where $\bm{s}\triangleq[s_1,s_2,\ldots,s_T]^{\mathsf T}\in \mathbb{C}^{T}$ and $\bm{Z}\triangleq[\bm{z}_1,\bm{z}_2,\ldots,\bm{z}_T]\in \mathbb{C}^{N \times T}$. Then, the covariance matrix of $\bm{Y}$ can be written as
	\begin{equation}
		\bm{R}_{\bm{Y}} = \frac{1}{T}\bm{Y}\bm{Y}^{\mathsf H} = \frac{1}{T}\bm{h}(\bm{x},u)\bm{s}^{\mathsf T}\bm{s}^{\mathsf *}\bm{h}(\bm{x},u)^{\mathsf H} + \sigma^2\bm{I}_N.
	\end{equation}
	Based on the procedures of the MUSIC algorithm, we can obtain the singular value decomposition (SVD) of $\bm{R}_{\bm{Y}}$ as
	\begin{equation}
		\bm{R}_{\bm{Y}} = \begin{bmatrix}
			\bm{u}_s, \bm{U}_z
		\end{bmatrix} \begin{bmatrix}
		\bm{\gamma}_s & \\ & \bm{\Gamma}_z
	\end{bmatrix} \begin{bmatrix}
	\bm{u}_s^{\mathsf H} \\ \bm{U}_z^{\mathsf H}
	\end{bmatrix},
	\end{equation}
	where $\bm{u}_s\in \mathbb{C}^{N}$ and $\bm{U}_z\in \mathbb{C}^{N\times (N-1)}$ are the singular vector(s) of the signal and noise subspaces, respectively, $\bm{\gamma}_s$ is the singular value of the signal subspace, and $\bm{\Gamma}_z\in \mathbb{R}^{(N-1)\times (N-1)}$ is a diagonal matrix with the singular values of the noise subspace on the diagonal. Since $\bm{\alpha}(\bm{x},u)$ is orthogonal to $\bm{U}_z$, and $\bm{\alpha}(\bm{x},\tilde{u})$ with $\tilde{u}\neq u$ is non-orthogonal to $\bm{U}_z$, i.e., $\bm{\alpha}(\bm{x},u)^{\mathsf H}\bm{U}_z\bm{U}_z^{\mathsf H}\bm{\alpha}(\bm{x},u) = 0$ and $\bm{\alpha}(\bm{x},\tilde{u})^{\mathsf H}\bm{U}_z\bm{U}_z^{\mathsf H}\bm{\alpha}(\bm{x},\tilde{u}) \neq 0$, there is a peak for the function $f(\bar{u})\triangleq\frac{1}{\bm{\alpha}(\bm{x},\bar{u})^{\mathsf H}\bm{U}_z\bm{U}_z^{\mathsf H}\bm{\alpha}(\bm{x},\bar{u})}$ at $\bar{u}=u$. Then, we can obtain the estimation of $u$ as \cite{shao2022target}
	\begin{equation}\label{MUSIC1D}
		\hat{u} = \arg\max_{\bar{u}\in [-1,1]} \frac{1}{\bm{\alpha}(\bm{x},\bar{u})^{\mathsf H}\bm{U}_z\bm{U}_z^{\mathsf H}\bm{\alpha}(\bm{x},\bar{u})},
	\end{equation}
	which can be solved by exhaustively searching for $\bar{u}$ over the interval $[-1,1]$. Then, the AoA estimation MSE can be expressed as
	\begin{align}
		{\rm{MSE}}(u)\triangleq\mathbb{E}\{|u-\hat{u}|^2\}.
	\end{align}
	Thus, the lower-bound of ${\rm{MSE}}(u)$, i.e., the CRB, is given by \cite{kay1993fundamentals}
	\begin{align}\label{CRB1D}
		 {\rm{MSE}}(u)\geq{\rm{CRB}}_u(\bm{x}) =  \frac{\sigma^2\lambda^2}{8\pi^2TPN|\beta|^2}\frac{1}{{\rm{var}}(\bm{x})},
	\end{align}
	where the variance function is defined as ${\rm{var}}(\bm{x})\triangleq\frac{1}{N}\sum_{n=1}^{N}x_n^2 - \mu(\bm{x})^2$ with $\mu(\bm{x})=\frac{1}{N}\sum_{n=1}^{N}x_n$ being the mean of vector $\bm{x}$. The result in \eqref{CRB1D} indicates that the CRB of AoA estimation MSE is dependent on the APV $\bm{x}$, i.e., ${\rm{CRB}}_u(\bm{x})$ decreases with the increase of ${\rm{var}}(\bm{x})$. Therefore, we can optimize APV $\bm{x}$ to maximize ${\rm{var}}(\bm{x})$ such that ${\rm{CRB}}_u(\bm{x})$ is minimized. Intuitively, in order to maximize ${\rm{var}}(\bm{x})$, the $N$ MAs should be spread apart as much as possible to increase the variation in $\bm{x}$, which yields an increased array aperture for synthesizing sensing beams with higher angular resolution.

	\subsection{CRB Minimization}
	In this subsection, we aim to minimize ${\rm{CRB}}_u(\bm{x})$ by optimizing $\bm{x}$. Based on the analysis of \eqref{CRB1D}, the objective function can be simplified as
	\begin{align}
		\min_{\bm{x}}{\rm{CRB}}_u(\bm{x}) \iff \max_{\bm{x}}{\rm{var}}(\bm{x}).
	\end{align}
	In order to mitigate the coupling among MAs at the receiver, a minimum separation distance $D$ is necessary for every antenna pair, i.e., $x_n-x_{n-1} \geq D$, $n=2,3,\ldots,N$. Then, the optimization problem for the APV of the MA array in the 1D line segment can be formulated as
	\begin{subequations}
		\begin{align}
			\textrm {(P1)}~~\max_{\bm{x}} \quad & {\rm{var}}(\bm{x})   \label{P1a}\\
			\text{s.t.} \quad & x_1\geq0, x_N\leq A, \label{P1b}\\
			& x_n-x_{n-1} \geq D,~~ n = 2,3,\ldots,N. \label{P1c}
		\end{align}
	\end{subequations}
	
	Note that an implicit condition for the feasibility of problem (P1) is $A\geq (N-1)D$ such that the minimum distance constraint \eqref{P1c} can be satisfied. Let $\bm{x}^\star\in\mathbb{R}^{N}$ denote an optimal solution for problem (P1). Although problem (P1) is a non-convex optimization problem due to the non-convexity of the objective function \eqref{P1a}, we can obtain $\bm{x}^\star$ by the following theorem.
	\begin{theorem}
		The optimal APV for minimizing the CRB of AoA estimation MSE in the 1D line segment is given by
		\begin{align}\label{1Doptimal}
			&x^\star_n=\left\{
			\begin{array}{ll}
				(n-1)D,    & n=1,2,\ldots,\lfloor N/2 \rfloor;\\
				A-(N-n)D,  & n=\lfloor N/2 \rfloor+1,\ldots,N.
			\end{array} \right.
		\end{align}
		The corresponding minimum CRB, given by \eqref{PAND} and \eqref{50}, decreases with $A$ in the order of $\mathcal{O}(A^{-2})$.
	\end{theorem}
	\begin{proof}
		See Appendix A.
	\end{proof}
	Theorem 1 demonstrates that to minimize the CRB of AoA estimation MSE in the 1D line segment, MAs should be partitioned into two groups as shown in Fig.~\ref{optimal_position_1D}. Specifically, half of the MAs are positioned in the leftmost extremity of the 1D line segment, while the remaining half are positioned in the rightmost extremity of the 1D line segment, both with the minimum inter-antenna spacing, $D$. Moreover, the CRB can be efficiently reduced by increasing the length of the line segment for antenna movement with a given number of MAs.
	
	\begin{figure}[!t]
		\centering
		\includegraphics[width=90mm]{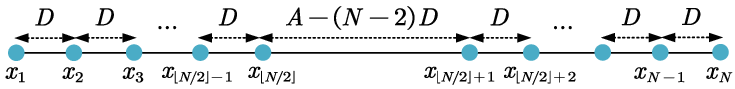}
		\caption{Illustration of the optimal MAs' positions for the 1D MA array.}
		\label{optimal_position_1D}
	\end{figure}
	
	\begin{figure}[!t]
		\centering
		\includegraphics[width=70mm]{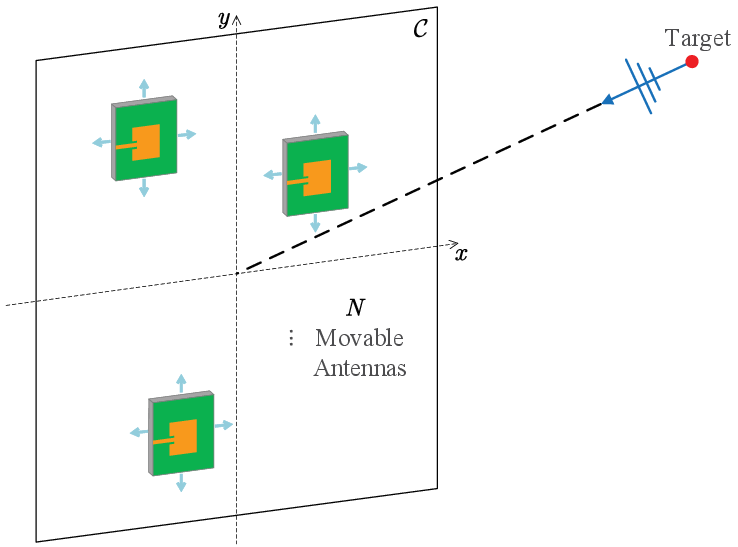}
		\caption{The 2D MA array for target angle estimation.}
		\label{system2}
	\end{figure}
	
	\begin{figure}[!t]
		\centering
		\includegraphics[width=85mm]{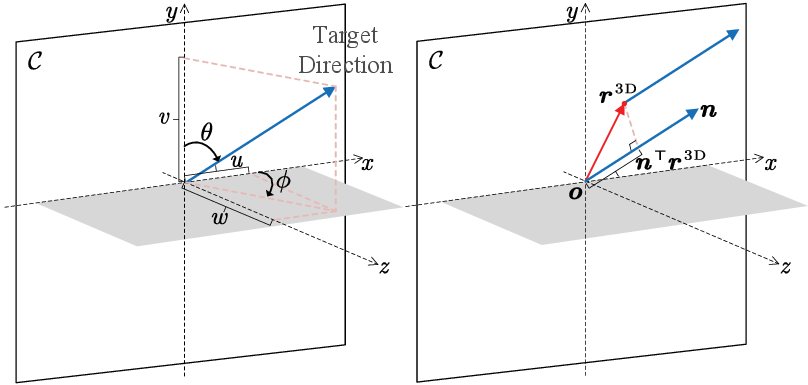}
		\caption{Illustration of the AoAs and geometrical parameters for the 2D MA array.}
		\label{AOA}
	\end{figure}

	\section{CRB Characterization for 2D MA Array}
	\subsection{System Model}
	As shown in Fig.~\ref{system2}, we consider in this section a wireless sensing system with $N$ MAs deployed on a 2D plane to estimate the target's AoAs with respect to $x$ and $y$ axes, respectively. The 2D region for antenna movement is assumed to be continuous and denoted by $\mathcal{C}$, where the 2D coordinate of the $n$th ($n=1,2,\ldots,N$) MA is denoted as $\bm{r}_n=[x_n,y_n]^{\mathsf T}\in\mathcal{C}$. Denote the collection of $N$ MAs' positions by $\tilde{\bm{r}}=\left[\bm{r}_1, \bm{r}_2, \ldots, \bm{r}_N\right] \in{\mathbb{R}^{2\times{N}}}$. As shown in Fig.~\ref{AOA}, the physical elevation and azimuth AoAs of the target-receiver LoS path are denoted by $\theta \in [0, \pi]$ and $\phi \in [0, \pi]$, respectively. For convenience, the spatial AoAs are defined as
	\begin{equation}
		u\triangleq\sin \theta \cos \phi, v\triangleq\cos \theta, w\triangleq\sin \theta \sin \phi.
	\end{equation}
	Then, the normalized wave vector of the target-receiver LoS path is expressed as $\bm{n} = [u, v, w]^{\mathsf T}$. Denote the three-dimensional (3D) coordinate vector of the origin of $\mathcal{C}$ and the MA by $\bm{o} = [0, 0, 0]^{\mathsf T}$ and $\bm{r}^{\rm 3D}=[x, y, 0]^{\mathsf T}$, respectively. Thus, the difference in signal propagation distance between MA's position $\bm{r}^{\rm 3D}$ and the reference point $\bm{o}$ can be represented as
	\begin{equation}\label{rho}
		\tau(\bm{r}) = \bm{n}^{\mathsf T}(\bm{r}^{\rm 3D}-\bm{o}) = x u + y v.
	\end{equation}
	Thus, the phase difference of the target-receiver LoS channel path between $\bm{r}^{\rm 3D}$ and $\bm{o}$ is given by $2\pi\tau(\bm{r})/\lambda$. The steering vector of the 2D MA array can thus be expressed as a function of the MAs' positions $\tilde{\bm{r}}$ and the two spatial AoAs $\bm{\eta}\triangleq[u,v]^{\mathsf T}$, i.e.,
	\begin{equation}\label{g2}
		\bm{\alpha}(\tilde{\bm{r}},\bm{\eta}) \triangleq \left[ e^{j\frac{2\pi}{\lambda}\tau(\bm{r}_1)}, e^{j\frac{2\pi}{\lambda}\tau(\bm{r}_2)}, \ldots, e^{j\frac{2\pi}{\lambda}\tau(\bm{r}_N)} \right]^{\mathsf T} \in{\mathbb{C}^{N}}.
	\end{equation}	
	As a result, the target-receiver LoS channel is given by
	\begin{equation}\label{H2}
		\bm{h}(\tilde{\bm{r}},\bm{\eta})=\beta\bm{\alpha}(\tilde{\bm{r}},\bm{\eta}).
	\end{equation}

	\subsection{AoA Estimation and CRB Minimization}
	Similar to the procedure for AoA estimation in the 1D MA array case, the joint estimation of $u$ and $v$ based on the MUSIC algorithm is given by \cite{shao2022target}
	\begin{equation}\label{MUSIC2D}
		\hat{\bm{\eta}} = \arg\max_{\bar{\bm{\eta}}\in[-1,1]\times[-1,1]} \frac{1}{\bm{\alpha}(\tilde{\bm{r}},\bar{\bm{\eta}})^{\mathsf H}\bm{U}_z\bm{U}_z^{\mathsf H}\bm{\alpha}(\tilde{\bm{r}},\bar{\bm{\eta}})},
	\end{equation}
	which can be solved by exhaustively searching for $\bar{\bm{\eta}}=[\bar{u},\bar{v}]^{\mathsf T}$ over the interval $[-1,1]\times[-1,1]$. Accordingly, for the MA-aided sensing system with 2D antenna movement region $\mathcal{C}$, the CRB of AoA estimation MSE is given by \cite{kay1993fundamentals}
	\begin{align}\label{CRBr}
		{\rm{MSE}}(u)\geq{\rm{CRB}}_u(\tilde{\bm{r}}) &=  \frac{\sigma^2\lambda^2}{8\pi^2TPN|\beta|^2}\frac{1}{{\rm{var}}(\bm{x})-\frac{{\rm{cov}}(\bm{x},\bm{y})^2}{{\rm{var}}(\bm{y})}}, \notag\\
		{\rm{MSE}}(v)\geq{\rm{CRB}}_v(\tilde{\bm{r}}) &=  \frac{\sigma^2\lambda^2}{8\pi^2TPN|\beta|^2}\frac{1}{{\rm{var}}(\bm{y})-\frac{{\rm{cov}}(\bm{x},\bm{y})^2}{{\rm{var}}(\bm{x})}},
	\end{align}
	where $\bm{x}\triangleq[x_1,x_2,\ldots,x_N]^{\mathsf T}\in \mathbb{R}^{N}$ and $\bm{y}\triangleq[y_1,y_2,\ldots,y_N]^{\mathsf T}\in \mathbb{R}^{N}$. The covariance function is defined as ${\rm{cov}}(\bm{x},\bm{y})\triangleq \frac{1}{N}\sum_{n=1}^{N}x_n y_n - \mu(\bm{x})\mu(\bm{y})$.
	
	The results in \eqref{CRBr} indicate that the CRB of AoA estimation MSE is dependent on the MAs' positions, which influence both ${\rm{CRB}}_u(\tilde{\bm{r}})$ and ${\rm{CRB}}_v(\tilde{\bm{r}})$. Specifically, ${\rm{CRB}}_u(\tilde{\bm{r}})$ and ${\rm{CRB}}_v(\tilde{\bm{r}})$ decrease with the increase of ${\rm{var}}(\bm{x})-\frac{{\rm{cov}}(\bm{x},\bm{y})^2}{{\rm{var}}(\bm{y})}$ and ${\rm{var}}(\bm{y})-\frac{{\rm{cov}}(\bm{x},\bm{y})^2}{{\rm{var}}(\bm{x})}$, respectively. Therefore, we can optimize MAs' positions $\tilde{\bm{r}}$ to jointly minimize ${\rm{CRB}}_u(\tilde{\bm{r}})$ and ${\rm{CRB}}_v(\tilde{\bm{r}})$. Intuitively, in order to maximize ${\rm{var}}(\bm{x})-\frac{{\rm{cov}}(\bm{x},\bm{y})^2}{{\rm{var}}(\bm{y})}$ and ${\rm{var}}(\bm{y})-\frac{{\rm{cov}}(\bm{x},\bm{y})^2}{{\rm{var}}(\bm{x})}$, i.e., by maximizing ${\rm{var}}(\bm{x})$ and ${\rm{var}}(\bm{y})$ while minimizing ${\rm{cov}}(\bm{x},\bm{y})$, MAs should be spread apart as much as possible in both the $x$ and $y$ directions to increase the variations in $\bm{x}$ and $\bm{y}$, while ensuring that $\bm{x}$ and $\bm{y}$ are symmetric with respect to $x$ and $y$ axes, respectively, thereby minimizing ${\rm{cov}}(\bm{x},\bm{y})$. However, there is in general a trade-off between minimizing ${\rm{CRB}}_u(\tilde{\bm{r}})$ and ${\rm{CRB}}_v(\tilde{\bm{r}})$ due to the coupling between ${\rm{var}}(\bm{x})$/${\rm{var}}(\bm{y})$ and ${\rm{cov}}(\bm{x},\bm{y})$.
	
	Then, we aim to minimize the maximum CRBs of estimation MSE for $u$ and $v$ (i.e., achieve the min-max CRB) by optimizing $\tilde{\bm{r}}$. Based on \eqref{CRBr}, the objective function can be simplified as
	\begin{align}
		&\min_{\tilde{\bm{r}}}\max  ~ [{\rm{CRB}}_u(\tilde{\bm{r}}), {\rm{CRB}}_v(\tilde{\bm{r}})] \iff \\
		& \max_{\tilde{\bm{r}}}\min  ~ \left[{\rm{var}}(\bm{x})-\frac{{\rm{cov}}(\bm{x},\bm{y})^2}{{\rm{var}}(\bm{y})}, {\rm{var}}(\bm{y})-\frac{{\rm{cov}}(\bm{x},\bm{y})^2}{{\rm{var}}(\bm{x})}\right].\notag
	\end{align}
	Thus, the antenna position optimization problem for the 2D MA array can be formulated as
	\begin{subequations}
		\begin{align}
			\textrm {(P2)}~~\max_{\tilde{\bm{r}},\delta} \quad & \delta \label{P2a}\\
			\text{s.t.} \quad & {\rm{var}}(\bm{x})-\frac{{\rm{cov}}(\bm{x},\bm{y})^2}{{\rm{var}}(\bm{y})} \geq \delta, \label{P2b}\\
			&{\rm{var}}(\bm{y})-\frac{{\rm{cov}}(\bm{x},\bm{y})^2}{{\rm{var}}(\bm{x})} \geq \delta, \label{P2c}\\
			& \tilde{\bm{r}} \in \mathcal{C}, \label{P2d}\\
			& \|\bm{r}_k-\bm{r}_l\|_2 \geq D,~~ k,l = 1,2,\ldots,N,~~ k\neq l,\label{P2e}
		\end{align}
	\end{subequations}
	where \eqref{P2e} is the minimum inter-antenna distance constraint to avoid antenna coupling. It is challenging to solve problem (P2) optimally because the fractional constraints \eqref{P2b} and \eqref{P2c} as well as the minimum inter-antenna distance constraints in \eqref{P2e} are non-convex with respect to $\tilde{\bm{r}}$. Moreover, the coupling between $\bm{x}$ and $\bm{y}$ exacerbates the difficulty of solving problem (P2).

	\begin{figure}[!t]
		\centering
		\includegraphics[width=70mm]{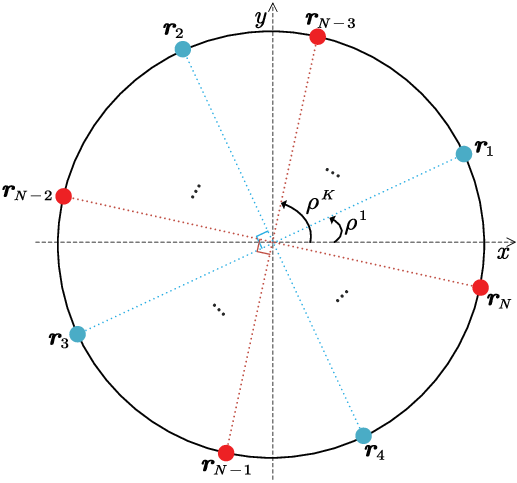}
		\caption{Illustration of the optimal MAs' positions in circular region $\mathcal{C}^\textrm{cir}$.}
		\label{optimal_position_circle}
	\end{figure}
	
	\subsection{Optimal Solution for Problem (P2) for Circular Region $\mathcal{C}^\textrm{cir}$}
	Although problem (P2) is difficult to solve for general antenna movement region $\mathcal{C}$, we show in this subsection that we can derive the optimal MAs' positions for the special case of circular region $\mathcal{C}^\textrm{cir}$ under certain numbers of MAs and circle radii. Specifically, we consider a circular region $\mathcal{C}^\textrm{cir}$ with its center at the coordinate origin and a radius of $A$, i.e., $\mathcal{C}^\textrm{cir}=\{(x,y)|x^2+y^2\leq A^2\}$. Then, we have the following theorem providing an  upper-bound for the objective value $\delta$ in problem (P2):
	\begin{theorem}
		For circular region $\mathcal{C}^\textrm{cir}$, the objective value $\delta$ in problem (P2) is upper-bounded by
		\begin{align}\label{21}
			\delta\leq\frac{A^2}{2},
		\end{align}
		and the corresponding lower-bound of the min-max CRB is
		\begin{align}
			\min_{\tilde{\bm{r}}}\max  ~ [{\rm{CRB}}_u(\tilde{\bm{r}}), {\rm{CRB}}_v(\tilde{\bm{r}})] \geq \frac{\sigma^2\lambda^2}{4\pi^2TPN|\beta|^2A^2},
		\end{align}
		which is tight when $D\leq 2A\sin(\pi/N)$ and $N=4K$, $K\in\mathbb{Z}^{+}$, with the corresponding optimal MAs' positions given by \eqref{xkyk}, \eqref{rho1}, and \eqref{rho2}.
	\end{theorem}
	\begin{proof}
		See Appendix B.
	\end{proof}
	Theorem 2 demonstrates that to minimize the CRB, it is desired to deploy MAs on the border of the region to maximize the array aperture as shown in Fig.~\ref{optimal_position_circle}. In order to satisfy the minimum inter-antenna distance constraints
	in \eqref{P2e} for every antenna pair on the border of the circle, it requires $D\leq 2A\sin(\pi/N)$. Moreover, when $N=4K$, $K\in\mathbb{Z}^{+}$, we can design $\bm{x}$ and $\bm{y}$ that are symmetric with respect to $x$ and $y$ axes, respectively, and thus ${\rm{cov}}(\bm{x},\bm{y})=0$. However, as the number of MAs increases, some MAs need to be deployed inside $\mathcal{C}^\textrm{cir}$ in order to satisfy the minimum inter-antenna distance constraints
	in \eqref{P2e}. Furthermore, the lower-bound of the min-max CRB is inversely proportional to the region size, i.e., $\pi A^2$. Therefore, the CRB can be efficiently reduced by increasing the size of the antenna movement region.

	\subsection{Performance Analysis for General 2D Region $\mathcal{C}$}
	\begin{figure}[!t]
		\centering
		\includegraphics[width=80mm]{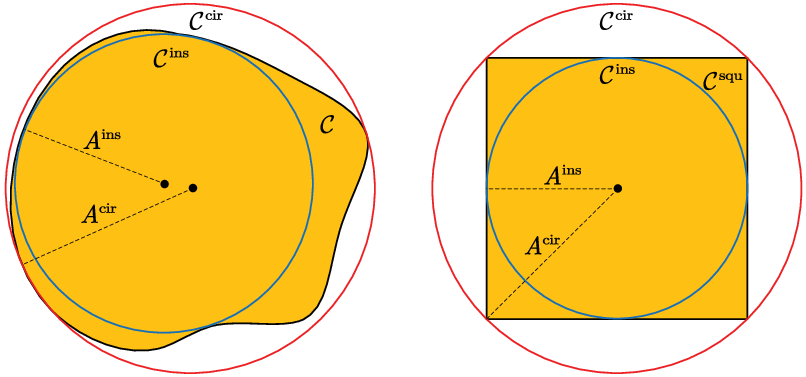}
		\caption{Illustration of the maximum inscribed circle and minimum circumscribed circle for an irregular 2D region, $\mathcal{C}$ (left) and a square region, $\mathcal{C}^\textrm{squ}$ (right).}
		\label{inscribed_circumscribed_circle}
	\end{figure}

	As shown in Fig.~\ref{inscribed_circumscribed_circle}, for a general 2D region $\mathcal{C}$, its minimum circumscribed circle and maximum inscribed circle can be determined using the sub-zone division approach \cite{huang2021effective}. Let $\mathcal{C}^\textrm{cir}$ and $\mathcal{C}^\textrm{ins}$ denote the regions covered by the minimum circumscribed circle and the maximum inscribed circle of $\mathcal{C}$, with radius $A^\textrm{cir}$ and $A^\textrm{ins}$, respectively. Since $\mathcal{C}^\textrm{ins}\subset\mathcal{C}\subset\mathcal{C}^\textrm{cir}$, the solution set of problem (P2) for $\mathcal{C}^\textrm{cir}$ contains that for region $\mathcal{C}$, while the solution set of problem (P2) for $\mathcal{C}$ contains that for region $\mathcal{C}^\textrm{ins}$. Denote the maximum objective value of problem (P2) and the corresponding min-max CRB for $\mathcal{C}$ by $\delta(\mathcal{C})$ and $\overline{{\rm{CRB}}}(\mathcal{C})$, respectively. Then, according to Theorem 2, the upper-bound of $\delta(\mathcal{C})$ can be expressed as
	\begin{align}\label{upperbound}
		\delta(\mathcal{C})\leq\delta(\mathcal{C}^\textrm{cir})\leq \frac{(A^\textrm{cir})^2}{2}.
	\end{align}
	Thus, the corresponding lower-bound of $\overline{{\rm{CRB}}}(\mathcal{C})$ is given by
	\begin{align}\label{lowerrbound}
		\overline{{\rm{CRB}}}(\mathcal{C}) \geq  \frac{\sigma^2\lambda^2}{4\pi^2TPN|\beta|^2(A^\textrm{cir})^2}.
	\end{align}
	Moreover, when $D\leq 2A^\textrm{ins}\sin(\pi/N)$ and $N=4K$ with $K\in\mathbb{Z}^{+}$, the upper-bound of the objective value $\delta$ in problem (P2) is tight for $\mathcal{C}^\textrm{ins}$, i.e., $\delta(\mathcal{C}^\textrm{ins})= (A^\textrm{ins})^2/2$. Thus,  we obtain the following lower-bound of $\delta(\mathcal{C})$ as
	\begin{align}
		\delta(\mathcal{C})\geq\delta(\mathcal{C}^\textrm{ins})= \frac{(A^\textrm{ins})^2}{2}.
	\end{align}
	Accordingly, the corresponding upper-bound of $\overline{{\rm{CRB}}}(\mathcal{C})$ is given by
	\begin{align}
		\overline{{\rm{CRB}}}(\mathcal{C}) \leq  \frac{\sigma^2\lambda^2}{4\pi^2TPN|\beta|^2(A^\textrm{ins})^2}.
	\end{align}
	For the typical square region $\mathcal{C}^\textrm{squ}$ with size $A\times A$, we have $A^\textrm{ins}=A/2$ and $A^\textrm{cir}=A/\sqrt{2}$. Thus, the upper-bound of $\delta(\mathcal{C}^\textrm{squ})$ is $A^2/4$, while the lower-bound of $\overline{{\rm{CRB}}}(\mathcal{C}^\textrm{squ})$ is $\frac{\sigma^2\lambda^2}{2\pi^2TPN|\beta|^2 A^2}$. Moreover, when $D\leq A\sin(\pi/N)$ and $N=4K$ with $K\in\mathbb{Z}^{+}$, the lower-bound of $\delta(\mathcal{C}^\textrm{squ})$ is $A^2/8$, while the upper-bound of $\overline{{\rm{CRB}}}(\mathcal{C}^\textrm{squ})$ is $\frac{\sigma^2\lambda^2}{\pi^2TPN|\beta|^2 A^2}$.

	\section{APV Optimization for 2D MA Array}
	Although the previous section provides lower- and upper-bounds for the min-max CRB in 2D antenna movement region $\mathcal{C}$ with arbitrary shapes, it is still challenging to optimize MAs' positions in problem (P2) due to the difficulty of solving it optimally. In this section, we introduce an alternating optimization algorithm to obtain locally optimal solutions for problem (P2). Specifically, the algorithm iteratively solves two subproblems of (P2) in an alternating manner, each for  optimizing one of the horizontal and vertical APVs $\bm{x}$ and $\bm{y}$ by keeping the other fixed. Next, we present this algorithm in detail and analyze its convergence and computational complexity.
	
	Notice that the analysis in the previous section holds for any continuous 2D antenna movement region $\mathcal{C}$ with arbitrary shapes. However, for ease of optimization, we consider in this section that $\mathcal{C}$ is a convex 2D region, such that the constraint \eqref{P2d} is convex. Moreover, we rewrite \eqref{P2b} and \eqref{P2c} in the standard quadratic form. Specifically, ${\rm{var}}(\bm{x})$, ${\rm{var}}(\bm{y})$, and ${\rm{cov}}(\bm{x},\bm{y})$ can be rewritten as
	\begin{align}
		{\rm{var}}(\bm{x}) &\triangleq \bm{x}^{\mathsf T} \bm{B} \bm{x}, \notag\\
		{\rm{var}}(\bm{y}) &\triangleq \bm{y}^{\mathsf T} \bm{B} \bm{y}, \notag\\
		{\rm{cov}}(\bm{x},\bm{y}) &\triangleq \bm{x}^{\mathsf T} \bm{B} \bm{y},
	\end{align}
	where $\bm{B}\triangleq \frac{1}{N}\bm{I}_N-\frac{1}{N^2}\bm{1}_N\bm{1}_N^{\mathsf T}$ is a positive semi-definite (PSD) matrix. Then, \eqref{P2b} and \eqref{P2c} can be equivalently transformed to
	\begin{subequations}
		\begin{align}
			& G(\bm{x},\bm{y})\triangleq G_1(\bm{x})-  G_2(\bm{x}, \bm{y})\triangleq\bm{x}^{\mathsf T}\bm{B}\bm{x}-\frac{\left(\bm{x}^{\mathsf T}\bm{B}\bm{y}\right)^2}{\bm{y}^{\mathsf T}\bm{B}\bm{y}} \geq \delta, \label{xBx}\\
			&G(\bm{y},\bm{x}) = G_1(\bm{y})-  G_2(\bm{y}, \bm{x}) = \bm{y}^{\mathsf T}\bm{B}\bm{y}-\frac{\left(\bm{y}^{\mathsf T}\bm{B}\bm{x}\right)^2}{\bm{x}^{\mathsf T}\bm{B}\bm{x}} \geq \delta, \label{yBy}
		\end{align}
	\end{subequations}
	where $G_1(\bm{x})\triangleq\bm{x}^{\mathsf T}\bm{B}\bm{x}$ and $G_2(\bm{x}, \bm{y})\triangleq \frac{\left(\bm{x}^{\mathsf T}\bm{B}\bm{y}\right)^2}{\bm{y}^{\mathsf T}\bm{B}\bm{y}}$.
	
	\subsection{Optimization of $\bm{x}$ with Given $\bm{y}$}
	In this subsection, we optimize $\bm{x}$ in problem (P2) with a given $\bm{y}$. Since constraints \eqref{xBx} and \eqref{yBy} are non-convex with respect to $\bm{x}$, we leverage the successive convex approximation (SCA) technique to transform them into convex constraints. Specifically, given $\bm{x}^p\in{\mathbb{R}^{N}}$ obtained in the $p$th iteration of SCA, since $G_1(\bm{x})$ is a convex function with respect to $\bm{x}$, it can be globally lower-bounded by the first-order Taylor expansion at $\bm{x}^p$ as
	\begin{align}\label{G_bar}
		&G_1(\bm{x})\geq \bar{G}_1(\bm{x}|\bm{x}^p)\\
		&\triangleq G_1(\bm{x}^p) + 2(\bm{x}^p)^{\mathsf T} \bm{B} (\bm{x}-\bm{x}^p) \notag\\
		&= 2(\bm{x}^p)^{\mathsf T} \bm{B} \bm{x} - G_1(\bm{x}^p). \notag
	\end{align}
	Moreover, given $\bm{y}$, $G_2(\bm{x}, \bm{y})$ is convex with respect to $\bm{x}$. Then, in the $p$th iteration of SCA, the convex surrogate function that globally lower-bounds $G(\bm{x},\bm{y})$ is given by
	\begin{align}
		G(\bm{x},\bm{y}) \geq \bar{G}_1(\bm{x}|\bm{x}^p) - G_2(\bm{x}, \bm{y}).
	\end{align}
	Thus, constraint \eqref{xBx} can be relaxed to be convex with respect to $\bm{x}$ as
	\begin{align}\label{30}
		\bar{G}_1(\bm{x}|\bm{x}^p) - G_2(\bm{x}, \bm{y}) \geq \delta.
	\end{align}
	
	On the other hand, given $\bm{y}$, \eqref{yBy} can be equivalently written as
	\begin{align}\label{loglog}
		\frac{\left(\bm{y}^{\mathsf T}\bm{B}\bm{x}\right)^2}{\bm{y}^{\mathsf T}\bm{B}\bm{y}-\delta} \leq \bm{x}^{\mathsf T}\bm{B}\bm{x}.
	\end{align}
	Notice that the left-hand side of \eqref{loglog} is jointly convex in $\{\bm{x}, \delta\}$ because $\left(\bm{y}^{\mathsf T}\bm{B}\bm{x}\right)^2$ is a convex quadratic function of $\bm{x}$, while $\bm{y}^{\mathsf T}\bm{B}\bm{y}-\delta$ is linear with respect to $\delta$. Recalling $G_1(\bm{x})=\bm{x}^{\mathsf T}\bm{B}\bm{x}$ and it is globally lower-bounded by $\bar{G}_1(\bm{x}|\bm{x}^p)$, constraint \eqref{loglog} can be relaxed as a convex constraint given by 
	\begin{align}\label{36}
		\frac{\left(\bm{y}^{\mathsf T}\bm{B}\bm{x}\right)^2}{\bm{y}^{\mathsf T}\bm{B}\bm{y}-\delta} \leq \bar{G}_1(\bm{x}|\bm{x}^p).
	\end{align}

	Furthermore, according to the Cauchy-Schwartz inequality, i.e., $\bm{u}^{\mathsf T}\bm{v}\leq\|\bm{u}\|_2\|\bm{v}\|_2$ for any two vectors $\bm{u}$ and $\bm{v}$ of equal size, the linear surrogate function that globally minorizes $\|\bm{r}_k-\bm{r}_l\|_2$ at $(\bm{r}_k^p-\bm{r}_l^p)$ can be constructed by letting $\bm{u}\leftarrow\bm{r}_k-\bm{r}_l$ and $\bm{v}\leftarrow\bm{r}_k^p-\bm{r}_l^p$, i.e., 
	\begin{align}
		&\|\bm{r}_k-\bm{r}_l\|_2\geq \frac{\left(\bm{r}_k^p-\bm{r}_l^p\right)^{\mathsf T}\left(\bm{r}_k-\bm{r}_l\right)}{\|\bm{r}_k^p-\bm{r}_l^p\|_2},
	\end{align}
	where $\bm{r}_n^p=[x_n^p,y_n]^{\mathsf T}$. Thus, constraint \eqref{P2e} can be relaxed as a linear constraint given by 
	\begin{align}\label{38}
		\frac{\left(\bm{r}_k^p-\bm{r}_l^p\right)^{\mathsf T}\left(\bm{r}_k-\bm{r}_l\right)}{\|\bm{r}_k^p-\bm{r}_l^p\|_2}\geq D,~~ k,l = 1,2,\ldots,N,~~ k\neq l.
	\end{align}
	 Hereto, in the $p$th iteration, the optimization of $\bm{x}$ can be relaxed as
	\begin{subequations}
		\begin{align}
			\textrm {(P3)}~~\max_{\bm{x},\delta} \quad & \delta \label{P4a}\\
			\text{s.t.} \quad & \eqref{30}, \eqref{36}, \eqref{P2d}, \eqref{38}, \notag
		\end{align}
	\end{subequations}
	which is a convex optimization problem since constraints \eqref{30}, \eqref{36}, \eqref{P2d}, and \eqref{38} are convex with respect to $\bm{x}$ and $\delta$. Thus, it can be efficiently solved using existing optimization toolboxes such as CVX \cite{grantcvx}.

	\subsection{Optimization of $\bm{y}$ with Given $\bm{x}$}
	In this subsection, we aim to optimize $\bm{y}$ with a given $\bm{x}$, where the non-convex constraints \eqref{xBx} and \eqref{yBy} are relaxed to be convex by using the SCA technique. Let $\bm{y}^q$ denote the $y$-axis APV in the $q$th iteration of SCA. In problem (P2), since the optimization of $\bm{y}$ with given $\bm{x}$ has a similar structure as the optimization of $\bm{x}$ with given $\bm{y}$, we can modify problem (P3) to optimize $\bm{y}$ by replacing $\left\{\bm{x}, \bm{x}^p, \bm{y} \right\}$ with $\left\{\bm{y}, \bm{y}^q, \bm{x} \right\}$. Thus, in the $q$th iteration of SCA, we can optimize $\bm{y}$ by solving the following optimization problem
	\begin{subequations}
		\begin{align}
			\textrm {(P4)}~~\max_{\bm{y},\delta} \quad & \delta \label{P5a}\\
			\text{s.t.} \quad & \bar{G}_1(\bm{y}|\bm{y}^q) - G_2(\bm{y}, \bm{x}) \geq \delta, \label{P5b}\\
			& \frac{\left(\bm{x}^{\mathsf T}\bm{B}\bm{y}\right)^2}{\bm{x}^{\mathsf T}\bm{B}\bm{x}-\delta} \leq \bar{G}_1(\bm{y}|\bm{y}^q), \label{P5c}\\
			& \frac{\left(\bm{r}_k^q-\bm{r}_l^q\right)^{\mathsf T}\left(\bm{r}_k-\bm{r}_l\right)}{\|\bm{r}_k^q-\bm{r}_l^q\|_2}\geq D,~~ k\neq l, \label{P5d}\\
			& \eqref{P2d}, \notag
		\end{align}
	\end{subequations}
	where $\bm{r}_n^q=[x_n,y_n^q]^{\mathsf T}$. Since constraints \eqref{P5b}, \eqref{P5c}, \eqref{P5d}, and \eqref{P2d} are convex with respect to $\bm{y}$ and $\delta$, problem (P4) is a convex optimization problem and can be efficiently solved using existing optimization toolboxes such as CVX \cite{grantcvx}.
	
	\subsection{Overall Algorithm}
	According to the solutions derived above for solving problems (P3) and (P4), we summarize the overall alternating optimization algorithm for solving problem (P2) in Algorithm~\ref{alg1}. Specifically, in lines 4-7, $\bm{x}$ is iteratively updated with a given $\bm{y}$ via solving problem (P3) until the increase of $\delta$ in \eqref{P4a} is below a predefined threshold $\epsilon_{\bm{x}}$. Then, in lines 9-12, $\bm{y}$ is iteratively updated with a given $\bm{x}$ via solving problem (P4) until the increase of $\delta$ in \eqref{P5a} is below a predefined threshold $\epsilon_{\bm{y}}$. The overall algorithm iterates between solving problems (P3) and (P4) until the increase of $\delta$ is below a predefined threshold $\epsilon$.
	
	\begin{algorithm}[!t]
		\caption{Alternating Optimization Algorithm for Solving Problem (P2)}
		\label{alg1}
		\begin{algorithmic}[1]
			\STATE \emph{Input:} $N$, $D$, $\epsilon$, $\epsilon_{\bm{x}}$, $\epsilon_{\bm{y}}$, $\bm{x}^0$, $\bm{y}^0$, $\mathcal{C}$.
			\STATE Initialization:  $p \leftarrow 0$, $q \leftarrow 0$, $\bm{y} \leftarrow \bm{y}^0$, $\delta \leftarrow 0$.
			
			\WHILE{Increase of $\delta$ is above $\epsilon$}	
			
			\WHILE{Increase of $\delta$ in \eqref{P4a} is above $\epsilon_{\bm{x}}$}
			\STATE Obtain $\bm{x}^{p+1}$ by solving problem (P3).
			\STATE $p \leftarrow p+1$.
			\ENDWHILE
			\STATE $\bm{x} \leftarrow \bm{x}^p$, $p \leftarrow 0$.
			
			\WHILE{Increase of $\delta$ in \eqref{P5a} is above $\epsilon_{\bm{y}}$}
			\STATE Obtain $\bm{y}^{q+1}$ by solving problem (P4).
			\STATE $q \leftarrow q+1$.
			\ENDWHILE
			\STATE $\bm{y} \leftarrow \bm{y}^q$, $q \leftarrow 0$.
			\STATE $\delta \leftarrow \min  \left[{\rm{var}}(\bm{x})-\frac{{\rm{cov}}(\bm{x},\bm{y})^2}{{\rm{var}}(\bm{y})},  {\rm{var}}(\bm{y})-\frac{{\rm{cov}}(\bm{x},\bm{y})^2}{{\rm{var}}(\bm{x})}\right]$.
			
			\ENDWHILE	
			\STATE \emph{Output:} $\bm{x}$, $\bm{y}$.
		\end{algorithmic}
	\end{algorithm}
	
	Next, we analyze the convergence of Algorithm~\ref{alg1} as follows. In the $p$th iteration for solving problem (P3), we have the following inequalities for constraint \eqref{xBx}:
	\begin{align}\label{convergence}
		G(\bm{x}^p,\bm{y}) &\overset{(a_1)}= \bar{G}_1(\bm{x}^p|\bm{x}^p) - G_2(\bm{x}^p, \bm{y}) \\
		&\overset{(a_2)}\leq \bar{G}_1(\bm{x}^{p+1}|\bm{x}^p) - G_2(\bm{x}^{p+1}, \bm{y}) \overset{(a_3)}\leq G(\bm{x}^{p+1},\bm{y}),\notag
	\end{align}
	where the equality $(a_1)$ holds because $G_1(\bm{x})= \bar{G}_1(\bm{x}|\bm{x}^p)$ at $\bm{x}=\bm{x}^p$; the inequality $(a_2)$ holds since we maximize $\bar{G}_1(\bm{x}|\bm{x}^p) - G_2(\bm{x}, \bm{y})$ in the $p$th iteration; the inequality $(a_3)$ holds because $\bar{G}_1(\bm{x}|\bm{x}^p)$ is a global lower-bound of $G_1(\bm{x})$. Hence, the sequence $\{G(\bm{x}^p,\bm{y})\}_{p=0}^{\infty}$ is monotonically increasing. Similarly, we have
	\begin{align}\label{convergence2}
		G(\bm{y},\bm{x}^p) &= G_1(\bm{y})-  \frac{\left(\bm{y}^{\mathsf T}\bm{B}\bm{x}^p\right)^2}{\bar{G}_1(\bm{x}^p|\bm{x}^p)} \\
		&\leq G_1(\bm{y})-  \frac{\left(\bm{y}^{\mathsf T}\bm{B}\bm{x}^{p+1}\right)^2}{\bar{G}_1(\bm{x}^{p+1}|\bm{x}^p)} \leq G(\bm{y},\bm{x}^{p+1}). \notag
	\end{align}
	Hence, the sequence $\{	G(\bm{y},\bm{x}^p)\}_{p=0}^{\infty}$ is monotonically increasing. Moreover, since $\delta = \min[G(\bm{x}^p,\bm{y}), G(\bm{y},\bm{x}^p)]$, the objective value of problem (P3) is guaranteed to increase monotonically. Similarly, the monotonic increment of the objective value after solving problem (P4) can be assured. Furthermore, alternately solving problems (P3) and (P4) guarantees the monotonically increasing objective value of problem (P2) during the outer iterations of Algorithm~\ref{alg1}, which is upper-bounded by \eqref{upperbound}. Therefore, the convergence of Algorithm~\ref{alg1} is guaranteed.
	
	The computational complexity of Algorithm~\ref{alg1} is analyzed as follows. The computational complexity for solving the convex optimization problem (P3) or (P4) is $\mathcal{O}(N^{3.5}\ln(1/\kappa))$ with accuracy given by $\kappa$ for the interior-point method \cite{fu2021reconf}. Let $I$, $I_{\bm{x}}$, and $I_{\bm{y}}$ denote the maximum number of iterations of executing lines 4–14 for alternately solving problems (P3) and (P4), that of executing lines 5-6 for solving problem (P3), and that of executing lines 10-11 for solving problem (P4), respectively. Therefore, the total computational complexity of Algorithm~\ref{alg1} is $\mathcal{O}(N^{3.5}\ln(1/\kappa)(I_{\bm{x}}+I_{\bm{y}})I)$.

	\section{Numerical Results}
	In this section, we present numerical results to evaluate the performance of our proposed solutions to optimize MAs' positions for target's AoA estimation with 1D or 2D MA arrays. The convergence thresholds in Algorithm~\ref{alg1} are set as $\epsilon=10^{-4}$ and $\epsilon_{\bm{x}}=\epsilon_{\bm{y}}=10^{-2}$. The minimum inter-antenna distance is set as $D=\lambda/2$. The average received SNR is defined as $P|\beta|^2/\sigma^2$. The number of snapshots is set as $T=1$.
	
	\subsection{1D MA Array}
	First, we consider the 1D line segment with length $A=10\lambda$. We set $N=16$, $\vartheta=45^{\circ}$, and thus $u=\cos \vartheta=0.71$. The considered benchmark schemes for setting antennas' positions are listed as follows: 1) \textbf{ULA with half-wavelength antenna spacing (ULAH)}: $\{x_n\}_{n=1}^N$ are set according to the ULA, with half-wavelength inter-antenna spacing; and 2) \textbf{ULA with full aperture (ULAF)}: $\{x_n\}_{n=1}^N$ are set according to the ULA with the largest achievable aperture, with the inter-antenna spacing of $A/(N-1)$, and $x_0=0$, $x_N=A$. Moreover, we adopt the MUSIC algorithm to estimate the AoA via \eqref{MUSIC1D} and compute the corresponding CRB of AoA estimation MSE via \eqref{CRB1D} for different schemes.	
	
	\begin{figure}[!t]
		\centering
		\includegraphics[width=75mm]{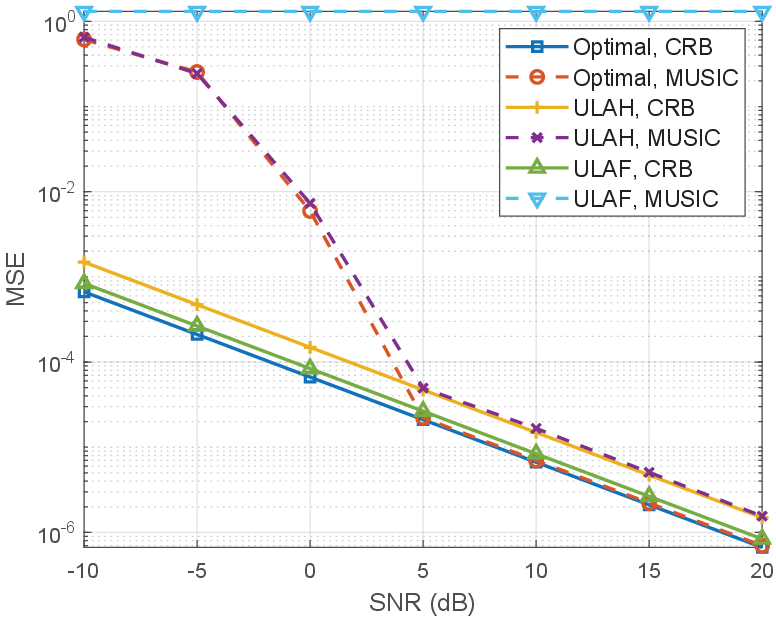}
		\caption{MSE versus SNR in the case of 1D array.}
		\label{1D_SNR}
	\end{figure}

	\begin{figure}[!t]
		\centering
		\includegraphics[width=75mm]{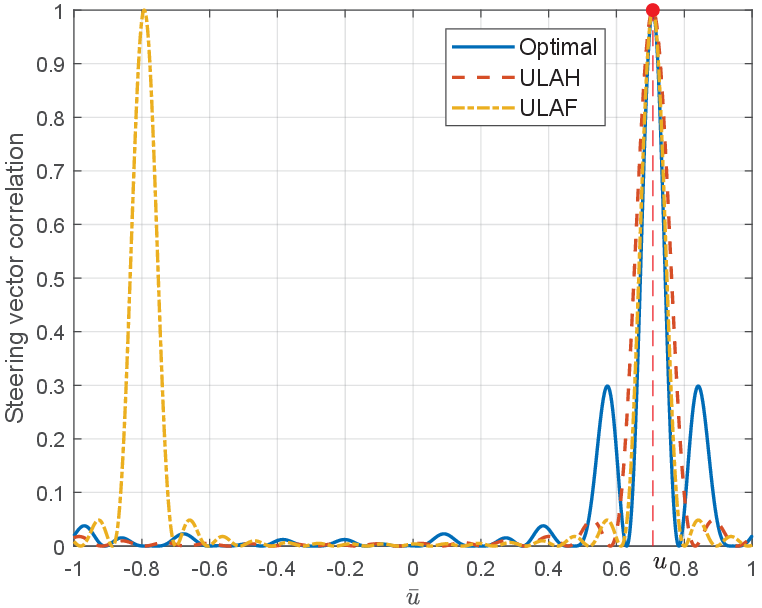}
		\caption{Comparison of steering vector correlation with different antennas' positions in the case of 1D array.}
		\label{1D_pattern}
	\end{figure}
	
	Fig.~\ref{1D_SNR} compares the AoA estimation MSEs (including both the actual MSE and its CRB) versus SNR for different schemes. It is shown that	the curves depicting AoA estimation MSE via the MUSIC algorithm can approach the CRB for both the optimal antennas' positions in \eqref{1Doptimal} and ULAH in the high-SNR regime. Furthermore, the proposed optimal antennas' positions in \eqref{1Doptimal} achieve a significantly lower MSE compared to the benchmark schemes. At $\text{SNR}=20$ dB, the proposed scheme demonstrates a $55.3\%$ MSE reduction over the ULAH scheme.

	To gain more insights, we illustrate in Fig.~\ref{1D_pattern} the steering vector correlation denoted by $q(\bar{u}|u)\triangleq\frac{1}{N^2}|\bm{\alpha}(\bm{x},u)^{\mathsf H} \bm{\alpha}(\bm{x},\bar{u})|^2$, for each case of antennas' positions versus $\bar{u}\in[-1,1]$ for a given $u=0.71$. To increase spatial resolution while minimizing ambiguity in AoA estimation, $q(\bar{u}|u)$ is desired to be an impulse function, i.e, $q(\bar{u}|u)\rightarrow\left\{
	\begin{array}{ll}
		1 & \bar{u}=u,\\
		0,  & \bar{u}\neq u.
	\end{array} \right.$ It is shown in Fig.~\ref{1D_pattern} that the optimal antennas' positions can yield a narrower main lobe compared to the benchmark schemes, which can achieve higher spatial resolution in AoA estimation, thus resulting in lower MSE. Additionally, it can be observed that $\frac{1}{N^2}|\bm{\alpha}(\bm{x},u)^{\mathsf H} \bm{\alpha}(\bm{x},-0.79)|^2=1$ for the ULAF scheme, indicating the ambiguity in distinguishing the true AoA $u=0.71$ from its false estimation of $-0.79$. This ambiguity results in the high MSE by the MUSIC algorithm for the ULAF scheme as shown in Fig.~\ref{1D_SNR}.

	\subsection{2D MA Array}
	
	\begin{figure}[!t]
		\centering
		\subfigure[$N=36$]{
			\begin{minipage}{.47\textwidth}
				\centering
				\includegraphics[scale=.6]{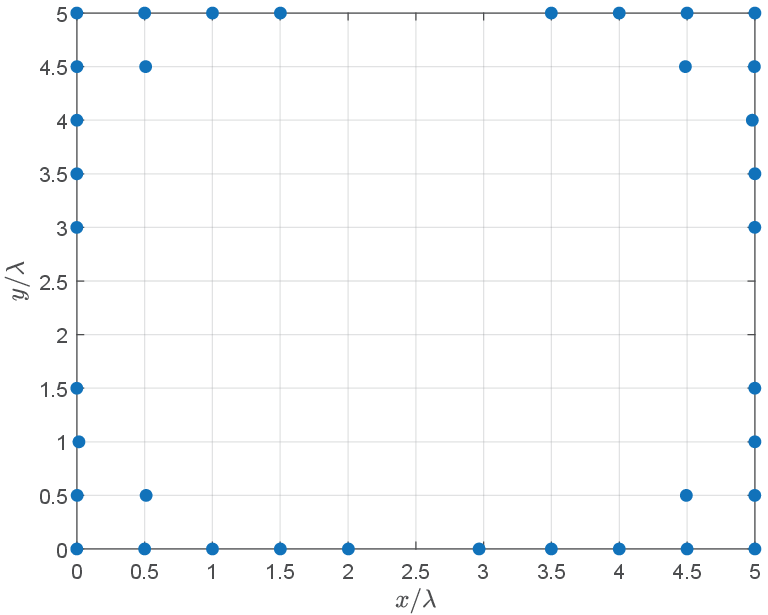}
			\end{minipage}
			\label{2D_position_36}
		}
		\subfigure[$N=100$]{
			\begin{minipage}{.47\textwidth}
				\centering
				\includegraphics[scale=.6]{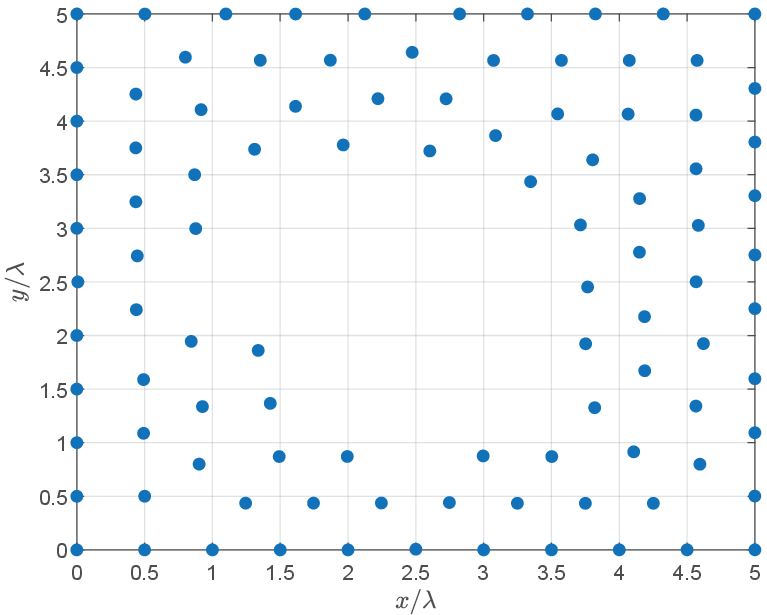}
			\end{minipage}
			\label{2D_position_100}
		}
		\caption{Illustration of the MAs' positions on 2D array.}
		\label{FIG4}
	\end{figure}

	\begin{figure}[!t]
		\centering
		\subfigure[${\rm{MSE}}(u)$]{
			\begin{minipage}{.47\textwidth}
				\centering
				\includegraphics[scale=.6]{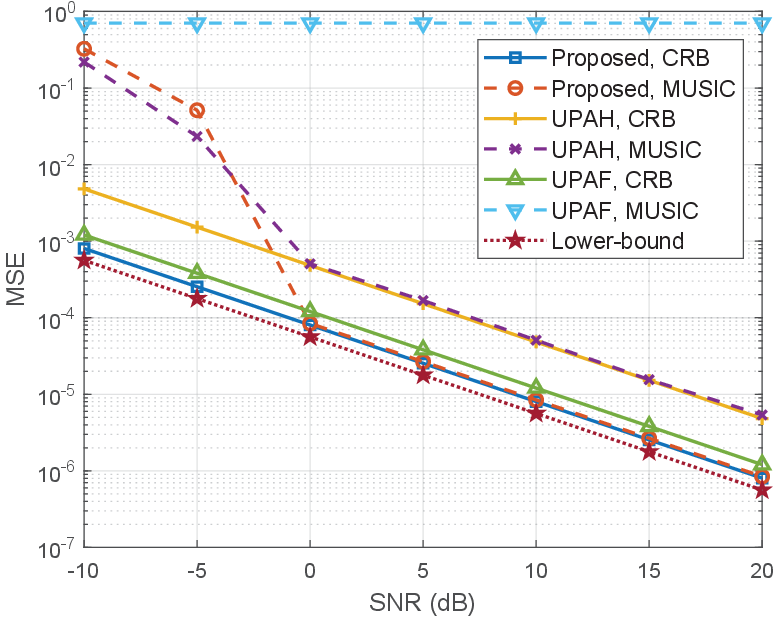}
			\end{minipage}
			\label{2D_SNR_u}
		}
		\subfigure[${\rm{MSE}}(v)$]{
			\begin{minipage}{.47\textwidth}
				\centering
				\includegraphics[scale=.6]{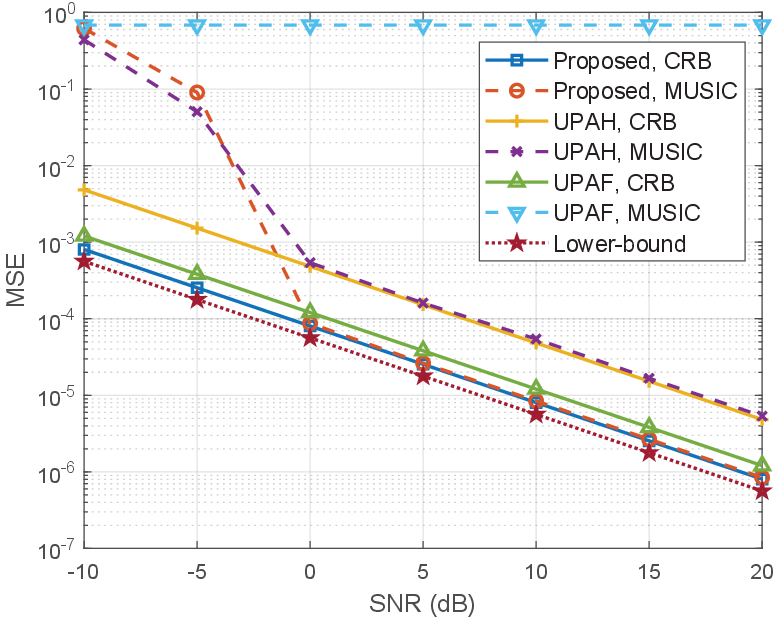}
			\end{minipage}
			\label{2D_SNR_v}
		}
		\caption{MSE versus SNR in the case of 2D array.}
		\label{FIG5}
	\end{figure}

	\begin{figure}[!t]
		\centering
		\subfigure[Proposed]{
			\begin{minipage}{.47\textwidth}
				\centering
				\includegraphics[scale=.6]{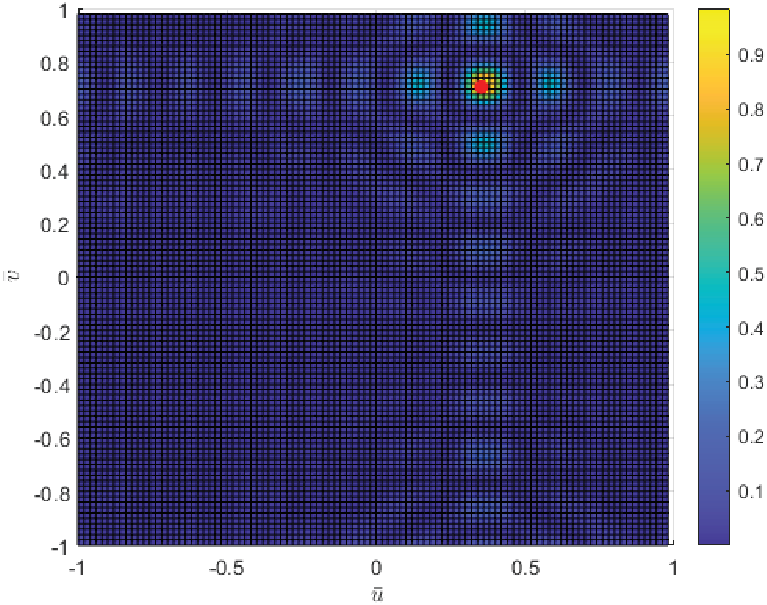}
			\end{minipage}
			\label{2D_pattern_proposed}
		}
		\subfigure[UPAH]{
			\begin{minipage}{.47\textwidth}
				\centering
				\includegraphics[scale=.6]{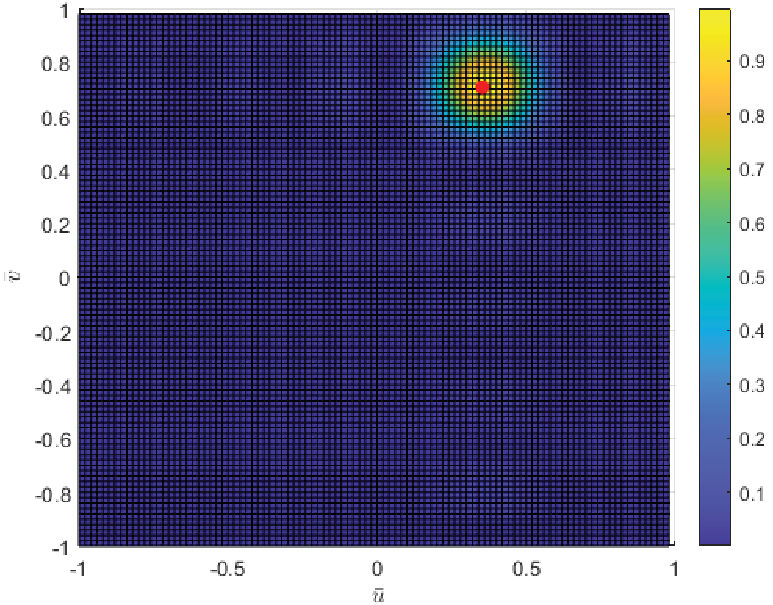}
			\end{minipage}
			\label{2D_pattern_UPAH}
		}
		\subfigure[UPAF]{
			\begin{minipage}{.47\textwidth}
				\centering
				\includegraphics[scale=.6]{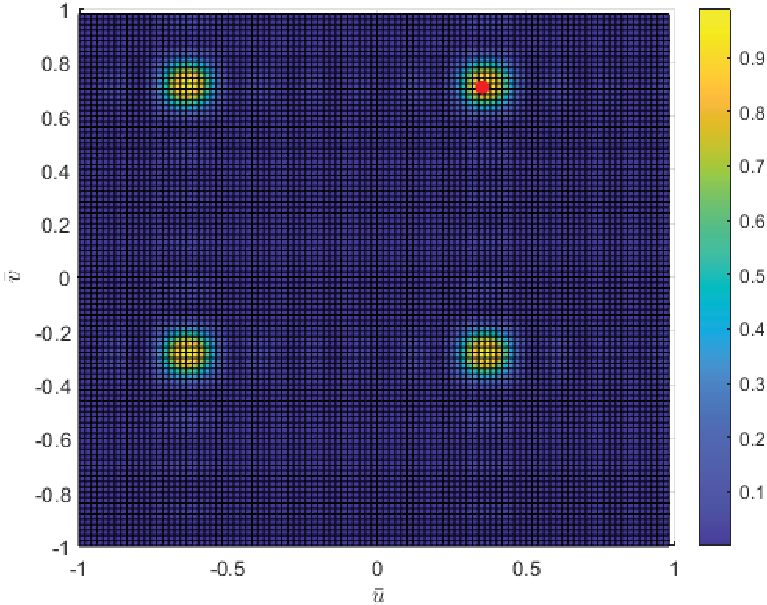}
			\end{minipage}
			\label{2D_pattern_UPAF}
		}
		\caption{Comparison of steering vector correlation with different antennas' positions in the case of 2D array.}
		\label{FIG6}
	\end{figure}

	\begin{figure}[!t]
		\centering
		\includegraphics[width=75mm]{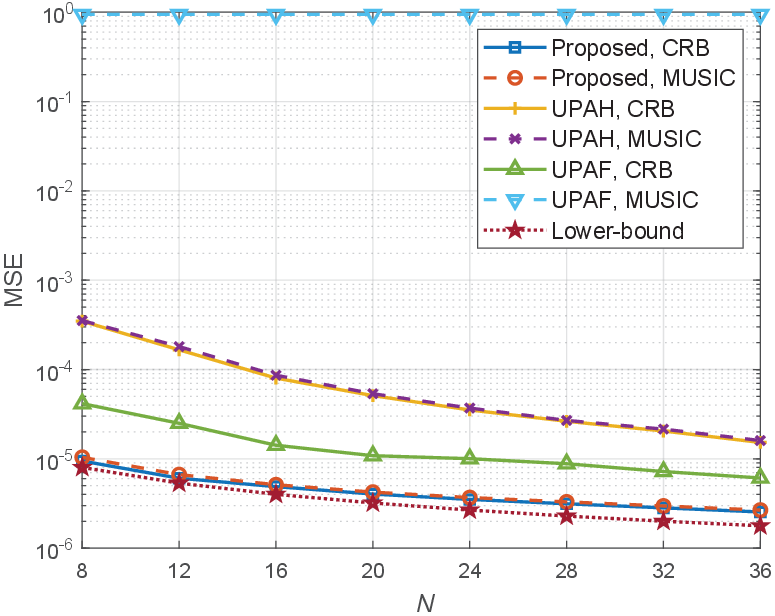}
		\caption{${\rm{MSE}}(u)$ versus $N$ in the case of 2D array.}
		\label{2D_N}
	\end{figure}
	
	Next, we present numerical results to validate the effectiveness of the proposed Algorithm~\ref{alg1} for AoA estimation with 2D MA array. We consider the 2D square region $\mathcal{C}^\textrm{squ}$ with size $A\times A$. We set $A=5\lambda$, $\theta=45^{\circ}$, $\phi=60^{\circ}$, and thus $u=\sin \theta \cos \phi=0.35$ and $v=\cos \theta=0.71$. The considered benchmark schemes for setting antennas' positions are listed as follows: 1) \textbf{UPA with half-wavelength antenna spacing (UPAH)}: $\{\bm{r}_n\}_{n=1}^N$ are set according to the UPA, with half-wavelength inter-antenna spacing both horizontally and vertically; and 2) \textbf{UPA with full aperture (UPAF)}: $\{\bm{r}_n\}_{n=1}^N$ are set according to the UPA with the largest achievable aperture, with the inter-antenna spacing of $A/\left(\left\lceil\sqrt{N}\right\rceil-1\right)$ both horizontally and vertically. For the proposed design for MAs' positions in Algorithm~\ref{alg1}, $\bm{x}^0$ and $\bm{y}^0$ are
	initialized based on the UPAF scheme. Moreover, we adopt the MUSIC algorithm to estimate the two spatial AoAs via \eqref{MUSIC2D} and compute the corresponding CRB of AoA estimation MSE via \eqref{CRBr} for different schemes. We also compute the lower-bound of min-max CRB for the considered square region $\mathcal{C}^\textrm{squ}$ via \eqref{lowerrbound}, i.e., $\frac{\sigma^2\lambda^2 }{2\pi^2TPN|\beta|^2 A^2}\triangleq \widetilde{{\rm{CRB}}}(\mathcal{C}^\textrm{squ})$.
	
	In Fig.~\ref{FIG4}, we illustrate the MAs' positions for the proposed algorithm (i.e., Algorithm~\ref{alg1}). We set $N=36$ and $N=100$, respectively. It can be observed that the optimized MAs' positions are symmetric with respect to both the $x$-axis and $y$-axis, for balancing the estimation accuracy of $u$ versus $v$. Additionally, the optimized MAs' positions are located as far from the center of the square region as possible to maximize the aperture of the MA array, thereby increasing the angular resolution for enhancing the AoA estimation performance.
	
	Fig.~\ref{FIG5} presents the comparison of the CRBs and actual values of ${\rm{MSE}}(u)$ and ${\rm{MSE}}(v)$ versus SNR for different schemes. We set $N=36$. It can be observed that the curves depicting AoA estimation MSE via the MUSIC algorithm can approach the CRB for both the proposed scheme and UPAH in the high-SNR regime. In addition, the proposed scheme achieves significantly lower MSE compared to the UPAH and UPAF schemes and its MSE is also close to the CRB lower-bound, $\widetilde{{\rm{CRB}}}(\mathcal{C}^\textrm{squ})$. Moreover, ${\rm{MSE}}(u)$ and ${\rm{MSE}}(v)$ show similar performance by using the proposed scheme, indicating its effectiveness for estimating both $u$ and $v$.

	To gain more insights, we illustrate in Fig.~\ref{FIG6} the steering vector correlation denoted by $\frac{1}{N^2}|\bm{\alpha}(\tilde{\bm{r}},\bm{\eta})^{\mathsf H} \bm{\alpha}(\tilde{\bm{r}},\bar{\bm{\eta}})|^2$, for each scheme versus $\bar{\bm{\eta}}\in[-1,1]\times[-1,1]$ with a given $\bm{\eta}=[0.35,0.71]^{\mathsf T}$. We set $N=36$. It can be observed that the proposed scheme exhibits a narrower main lobe compared to the benchmark schemes, resulting in the MSE reduction over the benchmark schemes. Additionally, it can be observed that $\frac{1}{N^2}|\bm{\alpha}(\tilde{\bm{r}},\bm{\eta})^{\mathsf H} \bm{\alpha}(\tilde{\bm{r}},\check{\bm{\eta}})|^2=1$ with $\check{\bm{\eta}}=[0.35,-0.28]^{\mathsf T}$, $[-0.64,0.28]^{\mathsf T}$, and $[-0.64,0.71]^{\mathsf T}$ for the UPAF scheme, indicating the ambiguity in distinguishing the true angles of $\bm{\eta}=[0.35,0.71]^{\mathsf T}$ from their false estimates of $\check{\bm{\eta}}$. This ambiguity results in the high MSE by the MUSIC algorithm for the UPAF scheme as shown in Fig.~\ref{FIG5}.
	
	In Fig.~\ref{2D_N}, we compare the CRBs and actual values of AoA estimation MSEs versus $N$ for different schemes. Given the similar performance between ${\rm{MSE}}(u)$ and ${\rm{MSE}}(v)$, we only show ${\rm{MSE}}(u)$ for brevity. We set $\text{SNR}=15$ dB. It is observed that the proposed scheme always outperforms the benchmark schemes and performs very closely to the CRB lower-bound, $\widetilde{{\rm{CRB}}}(\mathcal{C}^\textrm{squ})$, for different numbers of antennas. Furthermore, the proposed scheme demonstrates an MSE reduction of $97.1\%$ over the UPAH scheme, even with only $N=8$ MAs. This result indicates the practical significance of our proposed MA-aided sensing system design, even with a limited number of MAs.
	
	\section{Conclusions}
	In this paper, we proposed a new wireless sensing system equipped with MA arrays to improve the angle estimation accuracy by exploiting antenna position optimization. We derived the CRB for AoA estimation MSE as a function of the antennas' positions for both 1D and 2D MA arrays. First, for the case of 1D MA array, we obtained a globally optimal solution for the MAs' positions in closed form to minimize the CRB of AoA estimation MSE. Then, for the case of 2D MA array, we derived an optimal solution for the MAs' positions to achieve the min-max CRB for the special case of circular antenna movement region under certain numbers of MAs and circle radii. Thereby, both the lower- and upper-bounds of the min-max CRB were obtained for the antenna movement region with arbitrary shapes. Moreover, we developed an efficient alternating optimization algorithm to obtain a locally optimal solution for MAs' positions by iteratively optimizing one between their horizontal and vertical coordinates with the other being fixed. Numerical results demonstrated that our proposed 1D/2D MA array can significantly decrease the CRB of AoA estimation MSE as well as the actual MSE compared to conventional ULAs/UPAs with different values of uniform inter-antenna spacing. Furthermore, it was revealed that by optimizing MAs' positions, the steering vectors of our designed 1D/2D MA arrays exhibit low correlation in the angular domain, thus effectively reducing the ambiguity of angle estimation.

	\appendix
	\subsection{Proof of Theorem 1}
	For arbitrary $\bm{x}$ satisfying constraints \eqref{P1b} and \eqref{P1c}, we show that ${\rm{var}}(\bm{x}^\star)\geq {\rm{var}}(\bm{x})$ by sequentially adjusting MAs' positions $\{x_n\}_{n=1}^N$ through the following procedure.
	
	The procedure consists of $N$ steps, where one MA's position is adjusted in each step with those of all the other MAs being fixed. Let $\bm{x}^{(k-1)}\triangleq[x^{(k-1)}_1,x^{(k-1)}_2,\ldots,x^{(k-1)}_N]^{\mathsf T}$ denote the MAs' positions before the $k$th $(k=1,2,\ldots,N)$ adjustment, where $\bm{x}^{(0)}=\bm{x}$. Then, for $k=1,2,\ldots,\lfloor N/2 \rfloor$, the $k$th MA's position is updated by
	\begin{align}
		x^{(k)}_{k} \leftarrow \bm{x}^\star[k],
	\end{align}
	while the other $(N-1)$ MAs' positions remain unchanged, i.e.,
	\begin{align}
		x^{(k)}_{n} = x^{(k-1)}_{n}, n\in\mathcal{N}\setminus k,
	\end{align}
	where $\mathcal{N}\triangleq\{1,2,\ldots,N\}$. Similarly, for $k=\lfloor N/2 \rfloor+1,\lfloor N/2 \rfloor+2,\ldots,N$, the $(N-k+\lfloor N/2 \rfloor+1)$th MA's position is updated by
	\begin{align}
		x^{(k)}_{N-k+\lfloor N/2 \rfloor+1} \leftarrow \bm{x}^\star[N-k+\lfloor N/2 \rfloor+1],
	\end{align}
	while the other $(N-1)$ MAs' positions remain unchanged, i.e.,
	\begin{align}
		x^{(k)}_{n} = x^{(k-1)}_{n}, n\in\mathcal{N}\setminus (N-k+\lfloor N/2 \rfloor+1).
	\end{align}
	
	\begin{table}[!t]
		\centering
		\caption{Illustration of $\{\bm{x}^{(k)}\}_{k=1}^N$.}
		\label{Talbe1}
		\begin{tabular}{p{1.2cm}p{1.2cm}p{1.2cm}p{1.2cm}p{1.2cm}}
			\toprule
			$k$  &   $x^{(k)}_1$     &   $x^{(k)}_2$      &   $x^{(k)}_3$      &   $x^{(k)}_4$         \\
			\midrule
			0 ($\bm{x}$)  &   1  &   3  &   6  &   7.5     \\
			1 &   {\textbf{0}}  &   3  &   6  &   7.5     \\
			2  &   0  &   {\textbf{1}}  &   6  &   7.5    \\
			3  &   0  &   1  &   6  &   {\textbf{8}}     \\
			4 ($\bm{x}^\star$)  &   0  &   1  &   {\textbf{7}}  &   8     \\
			\bottomrule
		\end{tabular}
	\end{table}
	Table~\ref{Talbe1} shows an example of $\{\bm{x}^{(k)}\}_{k=1}^N$ for $\bm{x}=[1,3,6,7.5]^{\mathsf T}$ with $A=8$, $D=1$, and $N=4$, where $\bm{x}^\star$ is obtained by sequentially setting $x^{(1)}_1\leftarrow 0$, $x^{(2)}_2\leftarrow 1$, $x^{(3)}_4\leftarrow 8$, and $x^{(4)}_3\leftarrow 7$.
	
	In the above procedure, we can show that $\bm{x}^{(k)}$ satisfies constraints \eqref{P1b} and \eqref{P1c} based on mathematical induction. Specifically, since $\bm{x}^{(0)}=\bm{x}$, constraints \eqref{P1b} and \eqref{P1c} hold for $k=0$. Then, suppose that constraints \eqref{P1b} and \eqref{P1c} hold for $k\leq \lfloor N/2 \rfloor-1$. In the $(k+1)$th step, we have
	\begin{align}
		x^{(k+1)}_{k+1} &\leftarrow \bm{x}^\star[k+1],  \notag\\
		x^{(k+1)}_{n} &= x^{(k)}_{n}, n\in\mathcal{N}\setminus (k+1).
	\end{align}
	Then, we have
	\begin{align}
		x^{(k+1)}_{k+1}-x^{(k+1)}_{k} &= \bm{x}^\star[k+1]-x^{(k)}_{k} \notag\\ &=\bm{x}^\star[k+1]-\bm{x}^\star[k]\notag\\
		&= D.
	\end{align}
	Moreover, since 
	\begin{align}
		x^{(k)}_{k+1}-x^{(k)}_{k}&\geq D, \notag\\
		x^{(k)}_{k+2}-x^{(k)}_{k+1}&\geq D,
	\end{align}
	we have
	\begin{align}
		x^{(k+1)}_{k+2}-x^{(k+1)}_{k+1}&=\left(x^{(k+1)}_{k+2}-x^{(k+1)}_{k}\right)-\left(x^{(k+1)}_{k+1}-x^{(k+1)}_{k}\right) \notag\\
		&= \left(x^{(k)}_{k+2}-x^{(k)}_{k}\right)-D \notag\\
		&=\left(x^{(k)}_{k+2}-x^{(k)}_{k+1}\right)+\left(x^{(k)}_{k+1}-x^{(k)}_{k}\right)-D \notag\\
		&\geq D+D-D \notag\\
		&\geq D.
	\end{align}
	Moreover, since the other $(N-1)$ MAs' positions are fixed in the $(k+1)$th step, i.e., $x^{(k+1)}_{n} = x^{(k)}_{n}, n\in\mathcal{N}\setminus (k+1)$, we have $x^{(k+1)}_{n}-x^{(k+1)}_{n-1}\geq D, n=2,3,\ldots,N$. Since the antenna position adjustment procedure is symmetrical for $k=1,2,\ldots,\lfloor N/2 \rfloor$ and $k=\lfloor N/2 \rfloor+1,\lfloor N/2 \rfloor+2,\ldots,N$, we can show that constraints \eqref{P1b} and \eqref{P1c} hold for $k=\lfloor N/2 \rfloor+1,\lfloor N/2 \rfloor+2,\ldots,N$ in a similar manner. As such, the above induction ensures that $\bm{x}^{(k)}$ satisfies constraints \eqref{P1b} and \eqref{P1c} based on our designed position adjustment procedure.

	Next, we show that ${\rm{var}}(\bm{x}^{(k)})\geq {\rm{var}}(\bm{x}^{(k-1)})$ for $k=1,2,\ldots,N$. Define $\bar{\mu}(\bm{x},l)\triangleq\frac{1}{N-1}\sum_{n=1,n\neq l}^{N}x_n$. Accordingly, for $k=1,2,\ldots,\lfloor N/2 \rfloor$, we derive the difference between ${\rm{var}}(\bm{x}^{(k)})$ and ${\rm{var}}(\bm{x}^{(k-1)})$ as
	\begin{align}\label{varvar}
		&{\rm{var}}(\bm{x}^{(k)})-{\rm{var}}(\bm{x}^{(k-1)}) \\
		\overset{(b_1)}=&\frac{1}{N}\left(\left(x^{(k)}_k\right)^2-\left(x^{(k-1)}_k\right)^2\right)  - \frac{1}{N^2} \Big( \Big( (N-1)\bar{\mu}(\bm{x}^{(k)},k) \notag\\
		~~~~~~&+ x^{(k)}_k \Big)^2 	 - \left( (N-1)\bar{\mu}(\bm{x}^{(k-1)},k) + x^{(k-1)}_k \right)^2 \Big) \notag\\
		\overset{(b_2)}=&\frac{1}{N}\left(\left(x^{(k)}_k\right)^2-\left(x^{(k-1)}_k\right)^2\right) -  \frac{1}{N^2} \Big(\left(x^{(k)}_k\right)^2  \notag\\
		~~~~~~&- \left(x^{(k-1)}_k\right)^2  + 2(N-1)\bar{\mu}(\bm{x}^{(k)},k)\left(x^{(k)}_k-x^{(k-1)}_k\right)\Big) \notag\\
		=&\frac{N-1}{N^2}\left(x^{(k)}_k-x^{(k-1)}_k\right)\left( x^{(k)}_k+x^{(k-1)}_k - 2\bar{\mu}(\bm{x}^{(k)},k) \right),\notag
	\end{align}
	where the equality  $(b_1)$ holds since $x^{(k)}_{n} = x^{(k-1)}_{n}$ for $n\in\mathcal{N}\setminus k$ and $\mu(\bm{x}^{(k)}) = \frac{1}{N} \left( (N-1)\bar{\mu}(\bm{x}^{(k)},k) + x^{(k)}_k \right)$. The equality  $(b_2)$ holds because $\bar{\mu}(\bm{x}^{(k)},k)=\bar{\mu}(\bm{x}^{(k-1)},k)$. On one hand, since $x^{(k)}_{k} = \bm{x}^\star[k]=(k-1)D$ for $k=1,2,\ldots,\lfloor N/2 \rfloor$, we have
	\begin{align}\label{xkk1}
		x^{(k-1)}_k &= \sum_{n=2}^{k} \left( x^{(k-1)}_n-x^{(k-1)}_{n-1} \right) + x^{(k-1)}_1 \notag\\
		&\overset{(c)}\geq (k-1)D=x^{(k)}_{k},
	\end{align}
	where the inequality  $(c)$ holds since $\bm{x}^{(k-1)}$ satisfies constraints \eqref{P1b} and \eqref{P1c}. Therefore, we have $x^{(k)}_k-x^{(k-1)}_k\leq 0$. On the other hand, for $k=1,2,\ldots,\lfloor N/2 \rfloor$, $x^{(k)}_k+x^{(k-1)}_k - 2\bar{\mu}(\bm{x}^{(k)},k)$ can be further represented as
	\begin{align}\label{xxmu}
		&x^{(k)}_k+x^{(k-1)}_k - 2\bar{\mu}(\bm{x}^{(k)},k) \notag\\
		=&\frac{1}{N-1}\left( (N-1)\left( x^{(k)}_k+x^{(k-1)}_k \right) - 2\sum_{n=1,n\neq k}^{N}x^{(k)}_n \right) 
		\notag\\		
		=& \frac{1}{N-1}\Big( 2\sum_{n=1}^{k-1}   \left(x^{(k)}_k-x^{(k)}_n\right) + 2\sum_{n=k+1}^{N}  \left(x^{(k-1)}_k-x^{(k)}_n \right) \notag\\
		~~~~~~&+ ((N-1)-2(k-1))x^{(k)}_k \notag\\
		~~~~~~&+ ((N-1)-2(N-k))x^{(k-1)}_k \Big) \notag\\
		\overset{(d_1)}\leq& \frac{1}{N-1}\Bigg( 2\sum_{n=1}^{k-1}(k-n)D-2\sum_{n=k+1}^{N}(n-k)D   \notag\\
		~~~~~~&+ ((N-1)-2(k-1))x^{(k-1)}_k \notag\\
		~~~~~~&+ ((N-1)-2(N-k))x^{(k-1)}_k \Bigg) \notag\\
		=&\frac{1}{N-1}(k(k-1)-(N-k+1)(N-k))D \notag\\
		\overset{(d_2)}\leq&0,
	\end{align}
	where the inequality  $(d_1)$ holds since $x^{(k)}_k-x^{(k)}_n=(k-n)D$ for $k=1,2,\ldots,\lfloor N/2 \rfloor$ and $n=1,2,\ldots,k-1$, while $x^{(k-1)}_k-x^{(k)}_n=x_k-x_n\leq(k-n)D$ for $k=1,2,\ldots,\lfloor N/2 \rfloor$ and $n=k+1,k+2,\ldots,N$, and $x^{(k)}_{k}\leq x^{(k-1)}_k$ according to \eqref{xkk1}. The inequality  $(d_2)$ holds due to $k\leq N-k$ and $k-1\leq N-k+1$ for $k=1,2,\ldots,\lfloor N/2 \rfloor$.
	
	Substituting \eqref{xkk1} and \eqref{xxmu} into \eqref{varvar}, we can conclude that ${\rm{var}}(\bm{x}^{(k)})-{\rm{var}}(\bm{x}^{(k-1)})\geq0$ for $k=1,2,\ldots,\lfloor N/2 \rfloor$. Since the antenna position adjustment procedure is symmetrical over $k=1,2,\ldots,\lfloor N/2 \rfloor$ and $k=\lfloor N/2 \rfloor+1,\lfloor N/2 \rfloor+2,\ldots,N$, we can also show that ${\rm{var}}(\bm{x}^{(k)})-{\rm{var}}(\bm{x}^{(k-1)})\geq0$ for $k=\lfloor N/2 \rfloor+1,\lfloor N/2 \rfloor+2,\ldots,N$ in a similar manner. As such, for arbitrary $\bm{x}$ satisfying constraints \eqref{P1b} and \eqref{P1c}, we have ${\rm{var}}(\bm{x})= {\rm{var}}(\bm{x}^{(0)})\leq{\rm{var}}(\bm{x}^{(1)})\leq\ldots\leq{\rm{var}}(\bm{x}^{(N)})={\rm{var}}(\bm{x}^\star)$. Since ${\rm{var}}(\bm{x}^\star)\geq {\rm{var}}(\bm{x})$ holds for any feasible $\bm{x}$ that satisfies constraints \eqref{P1b} and \eqref{P1c},  $\bm{x}^\star$ is an optimal solution for problem (P1). 
	
	Furthermore, since $\bm{x}^\star=[0,D,2D,\ldots,({\lfloor N/2 \rfloor}-1)D, A-(N-{\lfloor N/2 \rfloor}-1)D, A-(N-{\lfloor N/2 \rfloor}-2)D, \ldots, A-D, A]^{\mathsf T}$ is a combination of two equal-distance vectors with 
	$\bar{N}\triangleq\lfloor N/2 \rfloor$ and $\tilde{N}\triangleq N - \lfloor N/2 \rfloor$ elements, respectively, ${\rm{var}}(\bm{x}^\star)$ can be further written as
	\begin{align}\label{PAND}
		&{\rm{var}}(\bm{x}^\star)=\frac{1}{N}\sum_{n=1}^{N}\left(x^\star_n\right)^2 - \mu(\bm{x}^\star)^2 \\
		&=\frac{1}{6N}\Big( \bar{N}(\bar{N}-1)(2\bar{N}-1) D^2  +  \tilde{N}(\tilde{N}-1)(2\tilde{N}-1)D^2 \notag\\
		&~~~~~ + 6\tilde{N}\tilde{A}A \Big) -\frac{1}{4N^2}\left((\bar{N}-1)\bar{N} D + (A+\tilde{A})\tilde{N}\right)^2 \notag\\
		&=\left\{
		\begin{array}{ll}
			\frac{1}{12} \big( 3A^2-3(N-2)DA & \\
			~~~~~~~~~~~~+ (N-2)(N-1)D^2 \big),    & N~\text{is even},\\
			\frac{(N-1)(N+1)}{12N^2} \big( 3A^2-3(N-2)DA & \\
			~~~~~~~~~~~~+ (N^2-3N+3)D^2 \big),  & N~\text{is odd},
		\end{array} \right. \notag\\
		&\triangleq p(A,N,D), \notag
	\end{align}
	where $\tilde{A}\triangleq A-(\tilde{N}-1)D$.
	Thus, the corresponding minimum CRB is given by
	\begin{align}\label{50}
		{\rm{CRB}}_u(\bm{x}^\star) =  \frac{\sigma^2\lambda^2}{8\pi^2TPN|\beta|^2}\frac{1}{p(A,N,D)}.
	\end{align}
	It can be observed from \eqref{PAND} that $p(A,N,D)$ is a quadratic function with respect to $A$ and $D$, and is a quadratic/monotonic quadratic function with respect to $N$ when $N$ is even/odd. Moreover, since it requires $A\geq (N-1)D$ such that the minimum distance constraint \eqref{P1c} can be satisfied, $p(A,N,D)$ increases with $A$ for $A\geq (N-1)D$ in the order of $\mathcal{O}(A^2)$ as $A\rightarrow\infty$, while it increases with $N$ for $2\leq N\leq A/D+1$ and decreases with $D$ for $0< D\leq A/(N-1)$. Since ${\rm{CRB}}_u(\bm{x}^\star)$ is inversely proportional to $p(A,N,D)$, we can conclude that ${\rm{CRB}}_u(\bm{x}^\star)$ decreases with $A$ for $A\geq (N-1)D$ in the order of $\mathcal{O}(A^{-2})$, while it decreases with $N$ for $2\leq N\leq A/D+1$ and increases with $D$ for $0\leq D\leq A/(N-1)$.	This thus completes the proof of Theorem 1.

	\subsection{Proof of Theorem 2}
	The objective function of problem (P2) can be further written as
	\begin{align}
		&\max_{\bm{x},\bm{y}}\min  ~ \left[{\rm{var}}(\bm{x})-\frac{{\rm{cov}}(\bm{x},\bm{y})^2}{{\rm{var}}(\bm{y})}, {\rm{var}}(\bm{y})-\frac{{\rm{cov}}(\bm{x},\bm{y})^2}{{\rm{var}}(\bm{x})}\right] \notag\\
		\leq& \max_{\bm{x},\bm{y}}\min  ~ \left[{\rm{var}}(\bm{x}), {\rm{var}}(\bm{y})\right] \notag\\
		\leq& \max_{\bm{x},\bm{y}} ~ \frac{1}{2}({\rm{var}}(\bm{x}) + {\rm{var}}(\bm{y})) \notag\\
		=& \max_{\bm{x},\bm{y}} ~\frac{1}{2}\left( \frac{1}{N}\sum_{n=1}^{N}x_n^2 - \mu(\bm{x})^2 + \frac{1}{N}\sum_{n=1}^{N}y_n^2 - \mu(\bm{y})^2 \right) \notag\\
		\leq& \max_{\bm{x},\bm{y}} ~\frac{1}{2}\left( \frac{1}{N}\sum_{n=1}^{N}x_n^2 + \frac{1}{N}\sum_{n=1}^{N}y_n^2 \right) \notag\\
		=& \max_{\bm{x},\bm{y}} ~\frac{1}{2N} \sum_{n=1}^{N}\left(x_n^2 + y_n^2\right)  \leq  \frac{A^2}{2}.
	\end{align}
	Thus, for the circular region $\mathcal{C}^\textrm{cir}$, the objective value of problem (P2) is upper-bounded by $A^2/2$, where the upper-bound is tight when
	\begin{subequations}\label{con}
		\begin{align}
			&{\rm{cov}}(\bm{x},\bm{y}) = 0, \label{con1}\\
			&{\rm{var}}(\bm{x}) = {\rm{var}}(\bm{y}),\label{con2}\\
			&\mu(\bm{x})=\mu(\bm{y})=0,\label{con3}\\
			&x_n^2 + y_n^2=A^2, n=1,2,\ldots,N.\label{con4}
		\end{align}
	\end{subequations}
	The constraint \eqref{con4} indicates that all MAs should be deployed on the border of the circular region $\mathcal{C}^\textrm{cir}$ to achieve the upper-bound of $\delta$. In order to ensure that the minimum distance for every antenna pair on the border of the circle is larger than $D$, i.e., satisfying constraint \eqref{P2e} and constraint \eqref{con4} simultaneously, it requires $D\leq 2A\sin(\pi/N)$.
	
	Due to the non-convexity of constraint \eqref{con}, it is challenging to design the MAs' positions that can achieve the upper-bound of $\delta$ for general $N\in\mathbb{Z}^{+}$. Nevertheless, we can construct the MAs' positions $\tilde{\bm{r}}^\textrm{cir}=\left[\bm{x}^\textrm{cir}, \bm{y}^\textrm{cir}\right]^{\mathsf T}\in \mathbb{R}^{2 \times N}$ that satisfy constraint \eqref{con} for $N=4K$, $K\in\mathbb{Z}^{+}$, shown as follows.
	
	\begin{figure}[!t]
		\centering
		\includegraphics[width=80mm]{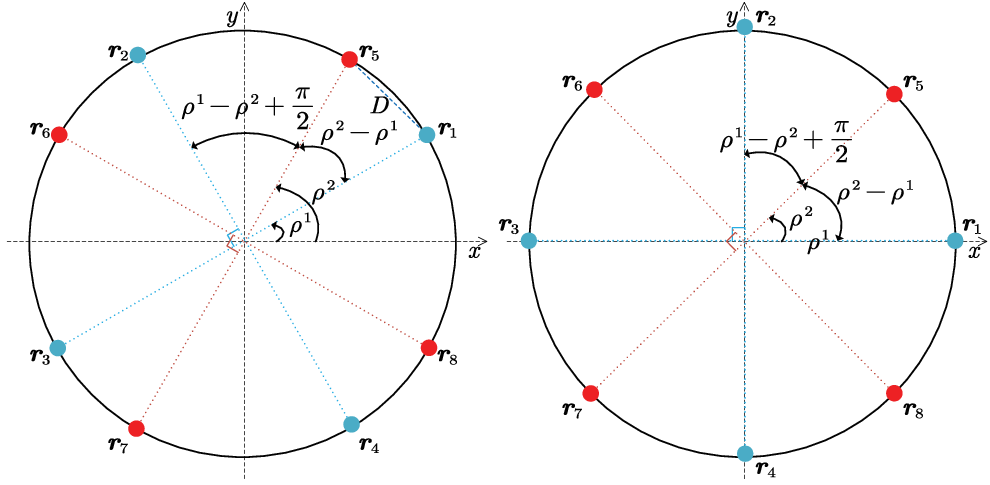}
		\caption{Illustration of the optimal MAs' positions for circular region $\mathcal{C}^\textrm{cir}$ with $\rho^1=\pi/6$, $\rho^2=\pi/3$ (left) or $\rho^1=0$, $\rho^2=\pi/4$ (right).}
		\label{optimal_position_circle_example}
	\end{figure}
	
	The MAs are first partitioned into $K$ groups, with each group containing four MAs. For the $k$th group ($k=1,2,\ldots,K$), denote the $x$-axis and $y$-axis APVs of the four MAs by $\bar{\bm{x}}^k=[\bar{x}^k_1,\bar{x}^k_2,\bar{x}^k_3,\bar{x}^k_4]^{\mathsf T}$ and $\bar{\bm{y}}^k=[\bar{y}^k_1,\bar{y}^k_2,\bar{y}^k_3,\bar{y}^k_4]^{\mathsf T}$, respectively. Consequently, $\bm{x}^\textrm{cir}$ and $\bm{y}^\textrm{cir}$ can be constructed as $\bm{x}^\textrm{cir}=\left[ (\bar{\bm{x}}^1)^{\mathsf T}, (\bar{\bm{x}}^2)^{\mathsf T}, \ldots, (\bar{\bm{x}}^K)^{\mathsf T} \right]^{\mathsf T}$ and $\bm{y}^\textrm{cir}=\left[ (\bar{\bm{y}}^1)^{\mathsf T}, (\bar{\bm{y}}^2)^{\mathsf T}, \ldots, (\bar{\bm{y}}^K)^{\mathsf T} \right]^{\mathsf T}$, respectively. As illustrated in Fig.~\ref{optimal_position_circle}, $\bar{\bm{x}}^k$ and $\bar{\bm{y}}^k$ are determined by the rotation angle $\rho^k$, which can be expressed as
	\begin{align}\label{xkyk}
		&\bar{x}^k_1 = A\cos(\rho^k), \bar{y}^k_1 = A\sin(\rho^k), \\
		&\bar{x}^k_2 = A\cos\left(\rho^k+\frac{\pi}{2}\right) = -\bar{y}^k_1, \bar{y}^k_2 = A\sin\left(\rho^k+\frac{\pi}{2}\right) = \bar{x}^k_1, \notag\\
		&\bar{x}^k_3 = A\cos\left(\rho^k+\pi\right)=-\bar{x}^k_1, \bar{y}^k_3 = A\sin\left(\rho^k+\pi\right)=-\bar{y}^k_1, \notag\\
		&\bar{x}^k_4 = A\cos\left(\rho^k+\frac{3\pi}{2}\right)=\bar{y}^k_1, \bar{y}^k_4 = A\sin\left(\rho^k+\frac{3\pi}{2}\right)=-\bar{x}^k_1.\notag
	\end{align}
	By substituting \eqref{xkyk} into \eqref{con}, we can verify that $\bar{\bm{x}}^k$ and $\bar{\bm{y}}^k$ satisfy constraint \eqref{con}. Thus, we have
	\begin{align}\label{conn}
		&\mu(\bm{x}^\textrm{cir})=\frac{4}{N}\sum_{k=1}^{K}\mu(\bar{\bm{x}}^k) = 0, \notag\\
		&\mu(\bm{y}^\textrm{cir})=\frac{4}{N}\sum_{k=1}^{K}\mu(\bar{\bm{y}}^k) = 0, \notag\\
		&{\rm{var}}(\bm{x}^\textrm{cir}) = \frac{4}{N}\sum_{k=1}^{K}{\rm{var}}(\bar{\bm{x}}^k) = \frac{4}{N}\sum_{k=1}^{K}{\rm{var}}(\bar{\bm{y}}^k) = {\rm{var}}(\bm{y}^\textrm{cir}), \notag\\
		&{\rm{cov}}(\bm{x}^\textrm{cir},\bm{y}^\textrm{cir}) = \frac{4}{N}\sum_{k=1}^{K}{\rm{cov}}(\bar{\bm{x}}^k,\bar{\bm{y}}^k) = 0.
	\end{align}
	Thus, the MAs' positions $\tilde{\bm{r}}^\textrm{cir}=\left[\bm{x}^\textrm{cir}, \bm{y}^\textrm{cir}\right]^{\mathsf T}$ satisfy constraint \eqref{con}, enabling them to achieve the upper-bound of the objective value of problem (P2). Without loss of generality, we assume
	\begin{align}\label{rho1}
		0\leq\rho^1<\rho^2<\ldots<\rho^K\leq\frac{\pi}{2},
	\end{align}
	such that $[\bar{x}^k_1,\bar{y}^k_1]^{\mathsf T}$ lies in the first quadrant of the $x$-$y$ coordinate system. Define $\bar{\bm{r}}^k_i\triangleq[\bar{x}^k_i,\bar{y}^k_i]^{\mathsf T}$, $i\in\mathcal{K}\triangleq\{1,2,3,4\}$. Furthermore, to satisfy the minimum distance constraint \eqref{P2e} of problem (P2), we need to ensure
	\begin{align}\label{rho2}
		\left\{
		\begin{array}{ll}
			\rho^{k}-\rho^{k-1}\geq2\arcsin\left(\frac{D}{2A}\right), ~k=2,3,\ldots,K,\\
			\rho^{1}-\rho^{K}+\frac{\pi}{2}\geq2\arcsin\left(\frac{D}{2A}\right),
		\end{array} \right.
	\end{align}
	such that $\|\bar{\bm{r}}^k_i-\bar{\bm{r}}^{k-1}_i\|_2 \geq D$ for $i\in\mathcal{K}$ and  $\|\bar{\bm{r}}^1_{i+1}-\bar{\bm{r}}^K_i\|_2 \geq D$ for $i\in\mathcal{K} \setminus 4$. Consequently, $\tilde{\bm{r}}^\textrm{cir}$ represents the optimal solution for problem (P2) in the circular region $\mathcal{C}^\textrm{cir}$. Notice that $\{\rho^k\}_{k=1}^K$ is not unique. For example, as shown in Fig.~\ref{optimal_position_circle_example}, we consider $N=8$, $K=2$, and $D=2A\sin(\pi/12)$ such that $\rho^{2}-\rho^{1}\geq \pi/6$ and $\rho^{1}-\rho^{2}+\pi/2\geq \pi/6$. Then, two possible solutions for $\{\rho^k\}_{k=1}^2$ are $\rho^1=\pi/6$, $\rho^2=\pi/3$ or $\rho^1=0$, $\rho^2=\pi/4$. This thus completes the proof of Theorem 2.

	\bibliographystyle{IEEEtran}
	\bibliography{IEEEabrv,IEEEexample}

\begin{thebibliography}{10}
\providecommand{\url}[1]{#1}
\csname url@samestyle\endcsname
\providecommand{\newblock}{\relax}
\providecommand{\bibinfo}[2]{#2}
\providecommand{\BIBentrySTDinterwordspacing}{\spaceskip=0pt\relax}
\providecommand{\BIBentryALTinterwordstretchfactor}{4}
\providecommand{\BIBentryALTinterwordspacing}{\spaceskip=\fontdimen2\font plus
\BIBentryALTinterwordstretchfactor\fontdimen3\font minus
  \fontdimen4\font\relax}
\providecommand{\BIBforeignlanguage}[2]{{%
\expandafter\ifx\csname l@#1\endcsname\relax
\typeout{** WARNING: IEEEtran.bst: No hyphenation pattern has been}%
\typeout{** loaded for the language `#1'. Using the pattern for}%
\typeout{** the default language instead.}%
\else
\language=\csname l@#1\endcsname
\fi
#2}}
\providecommand{\BIBdecl}{\relax}
\BIBdecl

\bibitem{jiang2021the}
W.~Jiang, B.~Han, M.~A. Habibi, and H.~D. Schotten, ``{The road towards 6G: A
  comprehensive survey},'' \emph{{IEEE} Open J. Commun. Soc.}, vol.~2, pp.
  334--366, Feb. 2021.

\bibitem{saad2020a}
W.~Saad, M.~Bennis, and M.~Chen, ``{A vision of 6G wireless systems:
  Applications, trends, technologies, and open research problems},''
  \emph{{IEEE} Netw.}, vol.~34, no.~3, pp. 134--142, May 2020.

\bibitem{chowdhury20206g}
M.~Z. Chowdhury, M.~Shahjalal, S.~Ahmed, and Y.~M. Jang, ``{6G wireless
  communication systems: Applications, requirements, technologies, challenges,
  and research directions},'' \emph{{IEEE} Open J. Commun. Soc.}, vol.~1, pp.
  957--975, 2020.

\bibitem{liu2022survey}
A.~Liu, Z.~Huang, M.~Li, Y.~Wan, W.~Li, T.~X. Han, C.~Liu, R.~Du, D.~K.~P. Tan,
  J.~Lu, Y.~Shen, F.~Colone, and K.~Chetty, ``{A survey on fundamental limits
  of integrated sensing and communication},'' \emph{{IEEE} Commun. Surveys
  Tuts.}, vol.~24, no.~2, pp. 994--1034, 2nd Quart., 2022.

\bibitem{shao2024intelligent}
X.~Shao, C.~You, and R.~Zhang, ``{Intelligent reflecting surface aided wireless
  sensing: Applications and design issues},'' \emph{{IEEE} Wireless Commun.},
  early access, 2024, doi: 10.1109/MWC.004.2300058.

\bibitem{mailloux2005phased}
R.~J. Mailloux, \emph{{Phased Array Antenna Handbook}}.\hskip 1em plus 0.5em
  minus 0.4em\relax 2nd ed. Norwood, MA, USA: Artech House, 2005.

\bibitem{wirth2005radar}
W.-D. Wirth, \emph{{Radar Techniques Using Array Antennas}}.\hskip 1em plus
  0.5em minus 0.4em\relax 2nd ed. Edison, NJ, USA: IET, 2005.

\bibitem{greene1978sparse}
C.~R. Greene and R.~C. Wood, ``{Sparse array performance},'' \emph{J. Acoust.
  Soc. Am.}, vol.~63, no.~6, pp. 1866--1872, Feb. 1978.

\bibitem{roberts2011sparse}
W.~Roberts, L.~Xu, J.~Li, and P.~Stoica, ``{Sparse antenna array design for
  MIMO active sensing applications},'' \emph{{IEEE} Trans. Antennas Propagat.},
  vol.~59, no.~3, pp. 846--858, Mar. 2011.

\bibitem{wang2023can}
H.~Wang and Y.~Zeng, ``{Can sparse arrays outperform collocated arrays for
  future wireless communications?}'' \emph{arXiv preprint arXiv:2307.07925},
  2023.

\bibitem{gazzah2009optimum}
H.~Gazzah and K.~Abed-Meraim, ``{Optimum ambiguity-free directional and
  omnidirectional planar antenna arrays for DOA estimation},'' \emph{{IEEE}
  Trans. Signal Process.}, vol.~57, no.~10, pp. 3942--3953, Oct. 2009.

\bibitem{zhu2023movablemagzine}
L.~Zhu, W.~Ma, and R.~Zhang, ``{Movable antennas for wireless communication:
  Opportunities and challenges},'' \emph{{IEEE} Commun. Mag.}, early access,
  2023, doi: 10.1109/MCOM.001.2300212.

\bibitem{zhu2024historical}
L.~Zhu and K.-K. Wong, ``{Historical review of fluid antenna and movable
  antenna},'' \emph{arXiv preprint arXiv:2401.02362}, 2024.

\bibitem{zhuravlev2015experi}
A.~Zhuravlev, V.~Razevig, S.~Ivashov, A.~Bugaev, and M.~Chizh, ``{Experimental
  simulation of multi-static radar with a pair of separated movable
  antennas},'' in \emph{Proc. {IEEE} COMCAS}, Tel Aviv, Israel, Nov. 2015, pp.
  1--5.

\bibitem{hinske2008using}
S.~Hinske and M.~Langheinrich, ``{Using a movable RFID antenna to automatically
  determine the position and orientation of objects on a tabletop},'' in
  \emph{Proc. EuroSSC}, Zurich, Switzerland, Oct. 2008, pp. 14--26.

\bibitem{zhu2022modeling}
L.~Zhu, W.~Ma, and R.~Zhang, ``{Modeling and performance analysis for movable
  antenna enabled wireless communications},'' \emph{{IEEE} Trans. Wireless
  Commun.}, early access, 2023, doi: 10.1109/TWC.2023.3330887.

\bibitem{mei2024movable}
W.~Mei, X.~Wei, B.~Ning, Z.~Chen, and R.~Zhang, ``{Movable-antenna position
  optimization: A graph-based approach},'' \emph{arXiv preprint
  arXiv:2403.16886}, 2024.

\bibitem{zhu2024performance}
L.~Zhu, W.~Ma, Z.~Xiao, and R.~Zhang, ``{Performance analysis and optimization
  for movable antenna aided wideband communications},'' \emph{arXiv preprint
  arXiv:2401.08974}, 2024.

\bibitem{zhu2023movable}
L.~Zhu, W.~Ma, B.~Ning, and R.~Zhang, ``{Movable-antenna enhanced multiuser
  communication via antenna position optimization},'' \emph{{IEEE} Trans.
  Wireless Commun.}, early access, 2023, doi: 10.1109/TWC.2023.3338626.

\bibitem{xiao2023multiuser}
Z.~Xiao, X.~Pi, L.~Zhu, X.-G. Xia, and R.~Zhang, ``{Multiuser communications
  with movable-antenna base station: Joint antenna positioning, receive
  combining, and power control},'' \emph{arXiv preprint arXiv:2308.09512},
  2023.

\bibitem{wu2023movable}
Y.~Wu, D.~Xu, D.~W.~K. Ng, W.~Gerstacker, and R.~Schober, ``{Movable
  antenna-enhanced multiuser communication: Optimal discrete antenna
  positioning and beamforming},'' in \emph{Proc. {IEEE} GLOBECOM}, Kuala
  Lumpur, Malaysia, Dec. 2023, pp. 7508--7513.

\bibitem{qin2024antenna}
H.~Qin, W.~Chen, Z.~Li, Q.~Wu, N.~Cheng, and F.~Chen, ``Antenna positioning and
  beamforming design for fluid antenna-assisted multiuser downlink
  communications,'' \emph{{IEEE} Wireless Commun. Lett.}, early access, 2024,
  doi: 10.1109/LWC.2024.3360117.

\bibitem{cheng2023sum}
Z.~Cheng, N.~Li, J.~Zhu, X.~She, C.~Ouyang, and P.~Chen, ``{Sum-rate
  maximization for movable antenna enabled multiuser communications},''
  \emph{arXiv preprint arXiv:2309.11135}, 2023.

\bibitem{yang2024flexible}
S.~Yang, W.~Lyu, B.~Ning, Z.~Zhang, and C.~Yuen, ``Flexible precoding for
  multi-user movable antenna communications,'' \emph{{IEEE} Wireless Commun.
  Lett.}, early access, 2024, doi: 10.1109/LWC.2024.3372569.

\bibitem{ma2022mimo}
W.~Ma, L.~Zhu, and R.~Zhang, ``{MIMO capacity characterization for movable
  antenna systems},'' \emph{{IEEE} Trans. Wireless Commun.}, vol.~23, no.~4,
  pp. 3392--3407, Apr. 2024.

\bibitem{chen2023joint}
X.~Chen, B.~Feng, Y.~Wu, D.~W.~K. Ng, and R.~Schober, ``{Joint beamforming and
  antenna movement design for moveable antenna systems based on statistical
  CSI},'' in \emph{Proc. {IEEE} GLOBECOM}, Kuala Lumpur, Malaysia, Dec. 2023,
  pp. 4387--4392.

\bibitem{ma2023compressed}
W.~Ma, L.~Zhu, and R.~Zhang, ``Compressed sensing based channel estimation for
  movable antenna communications,'' \emph{{IEEE} Commun. Lett.}, early access,
  2023, doi: 10.1109/LCOMM.2023.3310535.

\bibitem{xiao2024channel}
Z.~Xiao, S.~Cao, L.~Zhu, Y.~Liu, B.~Ning, X.-G. Xia, and R.~Zhang, ``{Channel
  estimation for movable antenna communication systems: A framework based on
  compressed sensing},'' \emph{{IEEE} Trans. Wireless Commun.}, early access,
  2024, doi: 10.1109/TWC.2024.3385110.

\bibitem{zhu2023movablebeam}
L.~Zhu, W.~Ma, and R.~Zhang, ``Movable-antenna array enhanced beamforming:
  Achieving full array gain with null steering,'' \emph{{IEEE} Commun. Lett.},
  vol.~27, no.~12, pp. 3340--3344, Dec. 2023.

\bibitem{ma2024multi}
W.~Ma, L.~Zhu, and R.~Zhang, ``Multi-beam forming with movable-antenna array,''
  \emph{{IEEE} Commun. Lett.}, vol.~28, no.~3, pp. 697--701, Mar. 2024.

\bibitem{hu2024secure}
G.~Hu, Q.~Wu, K.~Xu, J.~Si, and N.~Al-Dhahir, ``{Secure wireless communication
  via movable-antenna array},'' \emph{{IEEE} Signal Process. Lett.}, vol.~31,
  pp. 516--520, Jan. 2024.

\bibitem{tang2024secure}
J.~Tang, C.~Pan, Y.~Zhang, H.~Ren, and K.~Wang, ``{Secure MIMO communication
  relying on movable antennas},'' \emph{arXiv preprint arXiv:2403.04269}, 2024.

\bibitem{shao20246d}
X.~Shao, Q.~Jiang, and R.~Zhang, ``{6D movable antenna based on user
  distribution: Modeling and optimization},'' \emph{arXiv preprint
  arXiv:2403.08123}, 2024.

\bibitem{shao20246d2}
X.~Shao, R.~Zhang, Q.~Jiang, and R.~Schober, ``{6D movable antenna enhanced
  wireless network via discrete position and rotation optimization},''
  \emph{arXiv preprint arXiv:2403.17122}, 2024.

\bibitem{shao2022target}
X.~Shao, C.~You, W.~Ma, X.~Chen, and R.~Zhang, ``{Target sensing with
  intelligent reflecting surface: Architecture and performance},'' \emph{{IEEE}
  J. Sel. Areas Commun.}, vol.~40, no.~7, pp. 2070--2084, Jul. 2022.

\bibitem{kay1993fundamentals}
S.~M. Kay, \emph{{Fundamentals of Statistical Signal Processing: Estimation
  Theory}}.\hskip 1em plus 0.5em minus 0.4em\relax Prentice-Hall, Inc., 1993.

\bibitem{huang2021effective}
J.~Huang, R.~Yang, H.~Ge, and J.~Tan, ``{An effective determination of the
  minimum circumscribed circle and maximum inscribed circle using the subzone
  division approach},'' \emph{Meas. Sci. Technol.}, vol.~32, no.~7, pp. 1--9,
  May 2021.

\bibitem{grantcvx}
M.~Grant and S.~Boyd, \emph{{(Mar. 2014). CVX: MATLAB Software for Disciplined
  Convex Programming, Version 2.1}}, [Online]. Available: http://cvxr.com/cvx.

\bibitem{fu2021reconf}
M.~Fu, Y.~Zhou, Y.~Shi, and K.~B. Letaief, ``{Reconfigurable intelligent
  surface empowered downlink non-orthogonal multiple access},'' \emph{{IEEE}
  Trans. Commun.}, vol.~69, no.~6, pp. 3802--3817, Jun. 2021.

\end{thebibliography}

\end{document}